\documentclass[11pt,a4paper]{article}

\usepackage[T1]{fontenc}
\usepackage[utf8]{inputenc}
\usepackage{lmodern}

\usepackage[margin=2.5cm]{geometry}
\usepackage{amsmath}
\usepackage{amssymb}
\usepackage{amsthm}
\usepackage{mathtools}

\usepackage{booktabs}
\usepackage{array}
\usepackage{graphicx}
\graphicspath{{figures/}}
\usepackage{xcolor}

\usepackage[authoryear,round]{natbib}
\usepackage[colorlinks=true,linkcolor=blue,citecolor=blue,urlcolor=blue]{hyperref}
\usepackage{orcidlink}

\usepackage{macros}

\newcommand{\Lam}{\Lambda^{(s)}}            
\newcommand{\vphi}{\boldsymbol{\varphi}}     
\newcommand{\uzero}{\mathbf{u}_0}            
\newcommand{\uone}{\mathbf{u}_1}             
\newcommand{\Czero}{\mathbf{C}_0}            
\newcommand{\Cone}{\mathbf{C}_1}             
\newcommand{\vB}{\mathbf{B}}                 
\newcommand{\kap}{\kappa(\Phi)}              
\newcommand{\eff}{e(s)}                      
\newcommand{\thetaz}{\theta_0}               
\newcommand{\Ezero}{\mathbb{E}_0}
\newcommand{\Eone}{\mathbb{E}_1}
\newcommand{\Einf}{\mathbb{E}_\infty}
\newcommand{\Etau}{\mathbb{E}_\tau}
\newcommand{\Pzero}{\mathbb{P}_0}
\newcommand{\Pone}{\mathbb{P}_1}
\newcommand{\Pinf}{\mathbb{P}_\infty}
\newcommand{\Ptau}{\mathbb{P}_\tau}
\newcommand{\Dkl}{D_{\mathrm{KL}}}           
\newcommand{\tri}{\Delta}                    
\newcommand{\ARL}{\mathrm{ARL}_0}            
\newcommand{\phimax}{\varphi_{\max}}         
\newcommand{\Ncal}{N}                        

\newcommand{\SR}{\textsc{SR}}
\newcommand{\PFA}{\mathrm{PFA}}
\newcommand{\Rn}{R_n}
\newcommand{\Sn}{S_n}
\newcommand{\Tclass}{\mathcal{T}(\Lam)}

\newcommand{\huzero}{\hat{\mathbf{u}}_0}     
\newcommand{\huone}{\hat{\mathbf{u}}_1}      
\newcommand{\hCzero}{\hat{\mathbf{C}}_0}     
\newcommand{\hCone}{\hat{\mathbf{C}}_1}      
\newcommand{\hmatF}{\hat{\matF}}             
\newcommand{\hvK}{\hat{\vK}}                 
\newcommand{\hvY}{\hat{\vY}}                 

\theoremstyle{plain}
\newtheorem{theorem}{Theorem}[section]
\newtheorem{proposition}[theorem]{Proposition}

\theoremstyle{definition}
\newtheorem{definition}[theorem]{Definition}
\newtheorem{assumption}[theorem]{Assumption}
\newtheorem{example}[theorem]{Example}
\theoremstyle{remark}
\newtheorem{remark}[theorem]{Remark}

\begin{document}
\title{Generalized Stochastic Approximation of the Log-Likelihood Ratio\\
for Robust Sequential Change-Point Detection}
\author{Serhii Zabolotnii\,\orcidlink{0000-0003-0242-2234}\\[0.4em]
  \normalsize Cherkasy State Business College, Cherkasy 18028, Ukraine\\
  \normalsize State Scientific Research Institute of Armament and Military Equipment\\
  \normalsize Testing and Certification, Cherkasy, Ukraine\\
  \normalsize Uzhhorod National University, Uzhhorod, Ukraine}
\date{}
\maketitle

\begin{abstract}
Generalized stochastic approximation replaces the log-likelihood ratio of a sequential
change-point problem by an affine function of a fixed feature dictionary, whose coefficients
solve a normal system in the dictionary's moments under the pre- and post-change laws. The
increment is a rescaled orthogonal projection of the bounded score, and the information it
carries is bounded by the triangular discrimination between the laws. The dictionary decides what
is detectable: a symmetry shared by both laws annihilates a whole parity at every order, and one
further column raises the information by an explicit bordered-metric increment. A single
efficiency factor, at most one, governs all three procedures the increment drives ---
cumulative-sum, Shiryaev--Roberts and Shiryaev --- through false-alarm and delay bounds valid at
every threshold. Estimated coefficients are consistent and asymptotically normal. A Monte-Carlo
study measures the factor against the delay ratio it predicts, and locates the order at which the
dictionary stops paying.
\end{abstract}



\section*{Changes in this version}

This version replaces the manuscript by its theory. Versions~1--3 were built
around a benchmark study on nine public corpora, and the detection rates that
study reported are wrong. They came from a scoring path that recorded a run in
which the detector \emph{never raised an alarm} as a detection at the last
sample of the test segment, with a delay equal to the distance from the change
point to the end of that segment. The baseline detectors reported a miss in the
same situation and were scored correctly, so the error inflated the proposed
method alone.

Re-scored under the protocol those versions themselves stated --- a trial ends
in a false alarm if the detector stops before the change point, in a detection
if it stops after it, and in a missed target if no alarm is raised before the
end of the test segment --- the detection rates are:

\begin{center}
\begin{tabular}{@{}lcc@{}}
\toprule
Corpus & reported in v1--v3 & correct \\
\midrule
NSL-KDD (15 transitions) & $1.00$ & $0.33$ \\
NSL-KDD (30 transitions) & $1.00$ & $0.00$ \\
SKAB & $1.00$ & $0.02$ \\
TCPD & $0.80$ & $0.11$ \\
NAB~EC2 & $0.90$ & $0.65$ \\
NASA IMS & $1.00$ & $0.60$ \\
US RealInt & $1.00$ & $0.67$ \\
FEDFUNDS & $0.60$ & $0.20$ \\
FTSE~100 & $1.00$ & $1.00$ \\
PhysioNet~2019 & $0.00$ & $0.00$ \\
\bottomrule
\end{tabular}
\end{center}

The false-alarm rates are unaffected: a false alarm was determined by a stop
\emph{before} the change point, which the error does not touch. The two corpora
on which the detector raised no silent run, FTSE~100 and PhysioNet~2019, are
unchanged.

Three claims of v1--v3 do not survive the correction and are withdrawn: that the
detector is the only method that remains operative on extremely heavy-tailed
industrial vibration data; that it attains a false-alarm rate of zero at a
detection rate of one on network-intrusion data; and that it is the only method
reaching a detection rate of one at a false-alarm rate of zero on industrial
sensor data. On the heavy-tailed corpus a retrospective penalised-cost method
and a kernel method both detect every change point, and far sooner.

Two further defects were found while checking the first, and both are recorded
here rather than repaired in silence. Every basis value is clipped to a fixed
window, so on a series that is not on a unit scale the whole dictionary collapses
to one constant: the detector is then inert at every threshold rather than
conservative, which is what happened on sixty per cent of the trials of one
corpus. And at first order the three dictionaries compared throughout v1--v3
coincide identically, so the rows reporting them are one detector printed three
times.

What is \emph{not} affected: the theory, its proofs, and the Monte-Carlo study,
which is run on specified laws through a separate code path that scores a run
without an alarm correctly. That is what this version contains.

The benchmark study is not deleted. It is public in the code supplement named in
the data availability statement, with its result files, the script that re-scores
them, and a statement of what it does and does not establish. What it supports
is a scope result: the detector holds its false-alarm level where the change is
large relative to the in-control excursion of the statistic, and it does not see
a change of shape at an unchanged mean on a short calibration segment --- which
is the empirical face of the parity theorem proved here.

A correction issued earlier, and carried into this version, concerns
Theorem~2 of v1--v3: parts (a) and (c) are false and only monotonicity, part~(b),
holds. The complete-basis limit is not the Jeffreys divergence. The bound stated
in this version is the triangular discrimination, and it is proved.

This version also changes the normalisation of the normal system, and because that
changes printed constants it is spelled out here. Versions~1--3 fixed
$\mathbf F=\tfrac12(\mathbf C_0+\mathbf C_1)$ in the definition of
$\mathbf F\mathbf K=\mathbf Y$; this version uses $\mathbf F=\mathbf C_0+\mathbf C_1$,
the metric of the companion papers on moment-based discrimination, so that
statements quoted from them need no conversion --- and so that the definition
agrees with the inner product $\langle u,v\rangle_{\mathbf F}=\mathrm{Cov}_0(u,v)
+\mathrm{Cov}_1(u,v)$ that v1--v3 used in their own discussion of whitening,
which the factor $\tfrac12$ contradicted.

Relative to v1--v3 every coefficient vector $\mathbf K$, every increment
$\Lambda^{(s)}$ and every value of the information functional $J(s)$ is halved, and
every tilt root $\theta_0$ is doubled. Unchanged are the tilted increment
$\theta_0\Lambda^{(s)}$ that all three procedures consume, the efficiency factor,
the captured Fisher-information fraction and every stopping rule, so no bound and
no decision moves.

\emph{This is a change of units, not a further correction of Theorem~2.} The
erratum carried into v3 stated the complete-basis limit as
$(1/c)\cdot 2\Delta/(2-\Delta)$ for $\mathbf F=c(\mathbf C_0+\mathbf C_1)$, that is
$4\Delta/(2-\Delta)$ at $c=\tfrac12$; here $c=1$ and the same statement reads
$2\Delta/(2-\Delta)$. The Bernoulli witness of that erratum, $p_0=0.2$ against
$p_1=0.8$, accordingly prints $0.36/0.32=1.125$ here where v3 printed
$0.36/0.16=2.25$; the Jeffreys divergence it is compared against,
$1.2\ln 4\approx1.664$, does not depend on the normalisation, so this version also
gives the second witness, $p_0=0.01$ against $p_1=0.99$, where the ceiling
$0.9604/0.0198\approx48.5$ exceeds the symmetrised divergence
$1.96\ln 99\approx9.01$ --- the two stand in no fixed order, which is the point of
that part of the theorem.


\section{Introduction}\label{sec:intro}

\subsection{The problem}

Observations $x_1, x_2, \dots$ arrive one at a time. Up to an unknown time $\tau$ they are
independent with density $\fzero$; from $\tau$ on they are independent with density $\fone$. The
task is to raise an alarm soon after $\tau$ while keeping false alarms rare. When both densities
are known, the log-likelihood ratio $\ell = \ln(\fone/\fzero)$ is available at every observed
point, and the classical stopping rules built on it are optimal in the senses of
\citet{lorden1971procedures} and \citet{moustakides1986optimal}, of \citet{pollak1985optimal},
and of \citet{shiryaev1963optimum}.

This paper studies those same rules when the pair $(\fzero, \fone)$ is known only through a
finite list of generalized moments. A dictionary $\Phi = (\varphi_1, \dots, \varphi_s)$ of
features is fixed in advance and written as the vector
$\vphi = (\varphi_1, \dots, \varphi_s)\transp$; the pair enters only
through the means $\uzero = \Ezero\vphi$, $\uone = \Eone\vphi$ and the covariances $\Czero$,
$\Cone$, obtained as in Assumption~\ref{ass:h1}: $(\uzero, \Czero)$ from a calibration sample
recorded during normal operation, $(\uone, \Cone)$ either from a specification of the direction
and the size of the change in those moments or from a labelled sample of an earlier incident.
Generalized stochastic approximation maps that description to the increment
\[
\Lam(x) = \vK\transp\bigl(\vphi(x) - \tfrac12(\uzero + \uone)\bigr),
\qquad
\matF\vK = \vY, \quad \matF = \Czero + \Cone, \quad \vY = \uone - \uzero,
\]
whose normal system is that of \citet{kunchenko1997moment}, and to the
scalar $\Jfunc{s} = \vY\transp\matF^{-1}\vY$. No density need be named, and nothing is
estimated from the monitored stream. The regime this serves is a monitored stream whose law is
heavy-tailed or left unspecified, but whose feature moments are cheap to calibrate while it runs
in control. The object of the paper is the map from this description to
the operating characteristics of the rules driven by $\Lam$.

\subsection{Contributions}

\begin{itemize}
\item \emph{The approximation is a projection of the bounded score, not of the log-likelihood
  ratio.} In the space $L^2(m)$ of the average law $m = \tfrac12(\fzero + \fone)$, the increment
  $\Lam$ is a rescaled orthogonal projection of $\tanh(\ell/2)$ onto the span of the dictionary,
  and it is the minimiser of the moment criterion for testing of \citet{kunchenko1997moment}
  (Proposition~\ref{prop:ku1}). The information functional $\Jfunc{s}$ is non-decreasing under
  refinement of the dictionary and is bounded by a function of the triangular discrimination
  between $\fzero$ and $\fone$ --- not of their symmetrised divergence, from which it differs in
  both directions (Theorem~\ref{thm:info}). The identity is Proposition~1 of
  \citet{zabolotnii2026spl} and the envelope identity of \citet{zabolotnii2026mpfd}, in the
  same metric $\Czero+\Cone$; the ceiling is that envelope; its equality case, the limit and
  the counterexamples are proved here.
\item \emph{The structure of the dictionary decides what is detectable.} If $\fzero$ and $\fone$
  are symmetric about a common centre, every column odd about that centre contributes nothing to
  $\Jfunc{s}$; if $\fone$ is the reflection of $\fzero$ about that centre, every even column
  contributes nothing (Theorem~\ref{thm:parity}). A dictionary of the wrong parity is therefore
  blind at every order, and not only at the order at which the blindness is first noticed. Which
  columns are admissible is settled by the tails, and the clipping of
  Assumption~\ref{ass:bounded} is what removes that constraint
  (Remark~\ref{rem:admissible}). Adding one
  admissible column raises $\Jfunc{s}$ by an explicit bordered-metric increment, zero exactly when
  the new column carries no residual target mass (Proposition~\ref{prop:bordered}), so the
  monotonicity of Theorem~\ref{thm:info} has a computable price.
\item \emph{One efficiency factor governs the whole class of rules driven by the increment.}
  The equation $\Ezero e^{\theta\Lam} = 1$ has a unique positive root $\thetaz$, which exists
  because $\Lam$ is bounded and $\Ezero\Lam = -\Jfunc{s}/2 < 0$, and for the
  \CUSUM{}, the Shiryaev--Roberts rule and the Shiryaev rule all driven by $\thetaz\Lam$ there
  is a bound on the in-control run length from below, or on the probability of false alarm from
  above, at every threshold, an upper bound on the delay at every threshold, and one factor
  $\eff = \thetaz\Eone\Lam / \Dkl(\fone \Vert \fzero) \le 1$, where $\Dkl$ is the
  Kullback--Leibler divergence, with equality if and only if
  $\thetaz\Lam = \ell$ almost surely; in the local limit $\eff$ tends to the fraction $\kap$ of
  Fisher information captured by the dictionary (Theorem~\ref{thm:class}). The score-based
  \CUSUM{} of \citet{wu2023scusum} is the member of the class with $s = 1$ carrying the
  Hyv\"arinen-score difference as its single feature (Example~\ref{ex:score}).
\item \emph{The coefficients estimated from a calibration sample have a limit theory.} The
  coefficient vector, the information functional and the tilt root computed from a calibration
  sample of size $\Ncal$ are consistent and asymptotically normal, with a covariance obtained by
  the delta method; the captured information is lost at rate $O_p(\Ncal^{-1})$ and the efficiency
  at rate $O_p(\Ncal^{-1/2})$, because the normal system optimises a quadratic proxy of $\eff$
  and not $\eff$ itself (Theorem~\ref{thm:plugin}). A short Monte-Carlo study
  (Section~\ref{sec:numerics}) compares the predicted factor with measured delay ratios for two
  members of the class.
\end{itemize}

\subsection{What is not claimed}

\begin{remark}[Scope]\label{rem:notclaimed}
Four limitations are stated here and not revisited. (i)~No optimality of the stopping rules is
claimed: the optimality results of \citet{lorden1971procedures},
\citet{moustakides1986optimal}, \citet{pollak1985optimal} and \citet{shiryaev1963optimum} are
statements about $\ell$, whereas on $\Lam$ the excess delay is bounded by the computed factor
$1/\eff$, which is not minimised over any uncertainty class. (ii)~No optimality of the dictionary
is claimed: Theorem~\ref{thm:parity} and Proposition~\ref{prop:bordered} say when a column is
blind and what a column is worth, but which finite dictionary is best for a given pair is left
open. (iii)~Nothing is claimed about dependent observations: the change of measure behind
Theorem~\ref{thm:class} uses the martingale structure that independence supplies. (iv)~Nothing is
claimed about measured streams: the evidence here is the Monte-Carlo study of
Section~\ref{sec:numerics}, run on specified laws, and no comparison against competing detectors
on annotated corpora is offered or implied.
\end{remark}

\subsection{Related work}

When both laws are known the theory is settled. \citet{lorden1971procedures} formulated the
worst-case delay criterion under a constraint on the in-control average run length and obtained
the asymptotic lower bound $(\log\gamma)/\Dkl(\fone \Vert \fzero)$ on the delay at in-control run
length $\gamma$, and \citet{moustakides1986optimal} proved that
the \CUSUM{} attains it. \citet{pollak1985optimal} proposed the conditional criterion under which
the Shiryaev--Roberts rule is asymptotically optimal and \citet{pollak1987average} computed its
run lengths; \citet{shiryaev1963optimum} solved the Bayesian formulation with a geometric prior.
Book-length accounts are \citet{siegmund1985sequential} and \citet{tartakovsky2014sequential},
and the survey of \citet{xie2021sequential} is a recent entry point. Each of these procedures
requires $\ell$ pointwise.

A second line replaces the unknown post-change law by an estimate: nonanticipating estimation
\citep{lorden2005nonanticipating}, the information bounds and generalised likelihood ratio scheme
of \citet{lai1998information}, its window-limited form \citep{xie2023window}, and the multistream
schemes of \citet{mei2010efficient} and \citet{xie2013sequential}. A third line drops the
likelihood ratio: graph-based and nearest-neighbour statistics \citep{chen2015graph,
chen2019sequential} are distribution-free at the price of a slower response to weak changes, and
\citet{wu2023scusum} substitute the Hyv\"arinen score, so that an unnormalised model suffices.
The moment route is older. \citet{kunchenko1997moment} builds decision rules from a finite set of
moments of a chosen dictionary by minimising the ratio of the summed variances to the squared
separation of the means, the reciprocal of the deflection criterion whose reach in detection
\citet{picinbono1995deflection} delimited, and \citet{kunchenko2006stochastic} develops that
construction at length. This criterion is a testing apparatus, and it is not the polynomial
maximisation method of \citet{kunchenko2002polynomial}, an estimation apparatus whose coefficients
solve a normal system in cumulants for minimum asymptotic variance rather than for maximum
separation; the two are siblings under the same stochastic-polynomial umbrella, and only the first
is the ancestor of the increment studied here. \citet{palagin2010moment} adapted the testing
criterion to several hypotheses, and \citet{zabolotnii2015semiparametric} applied a polynomial
instance to sequential detection.

The present paper takes the moment description as given and asks what it determines.


\section{Setting and the class of procedures}\label{sec:setting}

\subsection{Model, metrics, assumptions}\label{sec:model}

Observations $x_1, x_2, \dots$ are real and independent. Under $\Pinf$ every $x_t$ has density
$\fzero$; under $\Ptau$, $\tau \ge 1$, the observations $x_t$ with $t < \tau$ have
density $\fzero$ and those with $t \ge \tau$ have density $\fone$. Expectations are written
$\Einf$ and $\Etau$; $\Ezero$, $\Eone$ (and $\Pzero$, $\Pone$) refer to a single observation
under $\fzero$, $\fone$. A detector is a stopping time $T$ of the observed sequence.

\begin{definition}[Performance metrics]\label{def:metrics}
The in-control average run length is $\ARL = \Einf T$. The worst-case detection delay is
\[
\ADD(T) = \sup_{\tau \ge 1}\ \operatorname*{ess\,sup}\
\Etau\bigl[(T - \tau + 1)^{+} \mid x_1, \dots, x_{\tau-1}\bigr].
\]
When the change time $\tau$ carries a prior distribution, the probability of false alarm is
$\PFA(T) = \Pr(T < \tau)$ under the joint law of the prior and the observations. In the
Monte-Carlo study of Section~\ref{sec:numerics} the change occurs at the first observation with
the statistic at its initial value, so the reported delay is $\ADD$ itself.
\end{definition}

The dictionary is an ordered list $\Phi = (\varphi_1, \dots, \varphi_s)$ of real functions of one
observation, written as the vector $\vphi(x) = (\varphi_1(x), \dots, \varphi_s(x))\transp$; $s$
is its order. Its moments under the two laws are $\uzero = \Ezero\vphi$, $\uone = \Eone\vphi$,
$\Czero = \operatorname{Cov}_0\vphi$ and $\Cone = \operatorname{Cov}_1\vphi$.

\begin{assumption}[Moments]\label{ass:moments}
Every feature has a finite second moment under $\fzero$ and under $\fone$; $\Czero$ is positive
definite; and $\fone$ vanishes wherever $\fzero$ does.
\end{assumption}

\begin{assumption}[Bounded dictionary]\label{ass:bounded}
There is a constant $\phimax < \infty$ with $|\varphi_i(x)| \le \phimax$ for all $i$ and all $x$.
\end{assumption}

Assumption~\ref{ass:bounded} is realised by clipping the raw features at $\phimax$, whose value
is a modelling choice of the analyst and is not fixed here. Under it every moment and every
exponential moment of the increment defined below is finite, which is what
Theorem~\ref{thm:class} uses; the clipping is a modelling device and not only a numerical
safeguard, since on a heavy-tailed $\fzero$ an unclipped dictionary of positive exponents need
not have finite second moments at all, and then Theorem~\ref{thm:class} has nothing to apply to.

\begin{assumption}[Post-change specification]\label{ass:h1}
The pair $(\uone, \Cone)$ is available in one of two modes.
\emph{(M) Specification.} The analyst supplies the target directly, as a vector of feature means
$\uone$ --- equivalently, as the direction and the size of the intended change
$\vY = \uone - \uzero$ --- together with a positive semidefinite $\Cone$, without naming a
density $\fone$.
\emph{(E) Estimation.} The pair $(\uone, \Cone)$ is the sample mean and sample covariance of the
features on a labelled post-change sample.
\end{assumption}

In neither mode is $\fone$ learned from the monitored stream. Mode~(M) fixes the alternative the
detector is tuned to and is the price of the construction: the guarantees below are statements
about the actual pair $(\fzero, \fone)$ carried by the stream, not about every law consistent
with a specified moment pair, and a target that points away from the actual change enters
Theorem~\ref{thm:class} through a smaller efficiency of that pair. The error a misspecified
target introduces does not vanish as the calibration sample grows, and is therefore not covered
by Theorem~\ref{thm:plugin}.

\begin{assumption}[Local family]\label{ass:local}
$\{g_\theta : \theta \ge 0\}$ is a family of densities with $g_0 = \fzero$, differentiable in
quadratic mean at $\theta = 0$ with score $q \in L^2(\fzero)$, $\Ezero q = 0$ and Fisher
information $I = \Ezero q^2 \in (0, \infty)$; the post-change law is $\fone = g_\theta$ for some
$\theta > 0$, and the local limit refers to $\theta \downarrow 0$.
\end{assumption}

Assumption~\ref{ass:local} is used only for the local limit of Theorem~\ref{thm:class}(e) and for
Definition~\ref{def:shape}; the rest of the paper fixes $\fzero$ and $\fone$ and does not embed
them in a family.

\subsection{The generalized stochastic approximation}\label{sec:gsa}

The log-likelihood ratio of one observation is
\begin{equation}\label{eq:llr}
\ell(x) = \ln \frac{\fone(x)}{\fzero(x)}.
\end{equation}
The construction replaces $\ell$ by the order-$s$ increment
\begin{equation}\label{eq:surrogate}
\Lam(x) = \vK\transp\bigl(\vphi(x) - \tfrac12(\uzero + \uone)\bigr)
= k_0 + \vK\transp\vphi(x),
\qquad k_0 = -\tfrac12\vK\transp(\uzero + \uone),
\end{equation}
whose coefficient vector solves the normal system
\begin{equation}\label{eq:normal-system}
\matF\vK = \vY, \qquad
\matF = \Czero + \Cone, \qquad
\vY = \uone - \uzero .
\end{equation}
Under Assumption~\ref{ass:moments}, with $\Cone$ positive semidefinite in either mode of
Assumption~\ref{ass:h1}, $\matF \succeq \Czero \succ 0$, so $\matF^{-1}$ exists and $\vK$
is well defined. The
information functional is
\begin{equation}\label{eq:Jfunc}
\Jfunc{s} = \vK\transp\vY = \vY\transp\matF^{-1}\vY,
\end{equation}
positive whenever $\vY \ne \mathbf{0}$. The metric $\Czero+\Cone$, the vector $\vY$ and the
functional $\Jfunc{s}$ are those of \citet{zabolotnii2026spl,zabolotnii2026mpfd}, so every
constant below is comparable with those papers without conversion. In which sense
\eqref{eq:normal-system} is the right system, and what $\Lam$ approximates, is the subject of
Proposition~\ref{prop:ku1}.

\paragraph{The procedure.} Given a calibration sample $x_1, \dots, x_{\Ncal}$ drawn from
$\fzero$, the procedure has four steps. (i) Evaluate the clipped features on the calibration
sample and take $\huzero$, $\hCzero$ to be their sample mean and sample
covariance; obtain $(\huone, \hCone)$ in the mode of
Assumption~\ref{ass:h1}. (ii) Solve $\hmatF\hvK = \hvY$ with
$\hmatF = \hCzero + \hCone$ and
$\hvY = \huone - \huzero$, which gives the estimated increment
$\hat{\Lam}$ of \eqref{eq:surrogate} with the estimated quantities in place of the population
ones. (iii) Solve
\[
\frac{1}{\Ncal}\sum_{j=1}^{\Ncal} \exp\bigl(\theta\,\hat{\Lam}(x_j)\bigr) = 1
\]
on the same calibration sample for its unique positive root $\hat\theta_0$, the sample version of
\eqref{eq:tilt}, and evaluate $\hat\theta_0\hat{\Lam}(x)$ at each new observation. (iv) Feed that
increment to one of the stopping rules of Section~\ref{sec:class}, with the threshold read off
the bound of Theorem~\ref{thm:class}(b) at the required false-alarm level. Steps (i)--(iii) are
\eqref{eq:surrogate}--\eqref{eq:Jfunc} and \eqref{eq:tilt} with sample moments, and the effect of
the substitution is Theorem~\ref{thm:plugin}; the order $s$ and the clipping level are inputs,
and the saturation criterion for $s$ is Proposition~\ref{prop:bordered}.

\subsection{The class of procedures}\label{sec:class}
Let $\thetaz$ be the unique positive root of \eqref{eq:tilt} below and $Z_n=\thetaz\Lam(x_n)$.
Three stopping rules are driven by $Z_n$:
\begin{align}
W_n&=\max(0,\,W_{n-1}+Z_n),\quad W_0=0, & T_h&=\inf\{n:\,W_n\ge h\},\label{eq:cusum}\\
\Rn&=(1+R_{n-1})\,e^{Z_n},\quad R_0=0, & T_A&=\inf\{n:\,\Rn\ge A\},\label{eq:sr}\\
\Sn&=\frac{(p+S_{n-1})\,e^{Z_n}}{1-p},\quad S_0=0, & T_B&=\inf\{n:\,\Sn\ge B\},\label{eq:shiryaev}
\end{align}
the \CUSUM{}, the Shiryaev--Roberts rule and the Shiryaev rule with a geometric prior
$\Pr(\tau=k)=p(1-p)^{k-1}$ on the change time. Write $\Tclass$ for the three. The scale
$\thetaz$ is what makes $e^{Z_n}$ a likelihood ratio (Theorem~\ref{thm:class}(a)).

\subsection{Dictionaries}\label{sec:dictionaries}

Every dictionary below is clipped at $\phimax$ (Assumption~\ref{ass:bounded}), and $c$ denotes a
centre fixed with the dictionary; $c = 0$ unless stated otherwise.

\emph{Polynomial}: $\varphi_i(x) = x^i$, $i = 1, \dots, s$. Its span contains the log-likelihood
ratio of a Gaussian mean change at every $s \ge 1$ and of a Gaussian variance change at every
$s \ge 2$, and its entries are the clipped moments of the calibration sample. The higher features
are the most sensitive to the tails: the sampling variance of the entry of $\hCzero$
involving $x^{2s}$ is governed by the kurtosis of $x^{s}$, not of $x$, and without clipping a
single large observation dominates the normal system.

\emph{Fractional}: $\varphi_i(x) = \operatorname{sgn}(x)|x|^{a_i}$ with exponents
$0 < a_i \le 1$, the linear feature being $a = 1$. Exponents below one damp large observations
and amplify small ones, and the dictionary is intended for heavy-tailed laws whose change is
concentrated near the centre; on such laws it can leave the normal system better conditioned than
the polynomial one, though that is a property of the pair and not of the dictionary alone. Every
column of it is odd about the origin, so by Theorem~\ref{thm:parity} it is blind to a change
between two laws symmetric about a common centre, at every order --- which is why the
mixed-parity dictionary below is needed.

\emph{Logarithmic}: $x$, $\ln|x|$, $x\ln|x|$, $(\ln|x|)^2$. Its span contains the log-likelihood
ratio of a scale change in the lognormal family, of a tail-index change in the Pareto family
--- for which $\ell(x) = \ln(\alpha_1/\alpha_0) - (\alpha_1 - \alpha_0)\ln x$ on the common
support --- and of a shape change in the gamma family, for all of which polynomial features of
any finite order converge slowly.

Two further dictionaries are needed for Theorem~\ref{thm:parity}, both from
\citet{zabolotnii2026mpfd} and both written here about the centre $c$. The \emph{mixed-parity}
dictionary of an exponent set $A$ is
$M(A) = \{x - c\} \cup \{\,|x - c|^{a},\ \operatorname{sgn}(x - c)|x - c|^{a} : a \in A\,\}$,
which pairs each exponent with a column of each parity about $c$. The \emph{half-line} dictionary
of two exponent sets is
$H(A^{+}, A^{-}) = \{x - c\} \cup \{\,((x - c)^{+})^{a} : a \in A^{+}\,\}
\cup \{\,((x - c)^{-})^{b} : b \in A^{-}\,\}$, where $y^{+} = \max(y, 0)$ and
$y^{-} = \max(-y, 0)$, which splits the support instead of mixing the parities; since
$(y^{\pm})^{a} = \tfrac12(|y|^{a} \pm \operatorname{sgn}(y)|y|^{a})$ and $\Jfunc{s}$ depends on
the dictionary only through its span, $M(A)$ and $H(A, A)$ carry the same information functional.

The regime the paper targets is defined through the local families of
Assumption~\ref{ass:local}.

\begin{definition}[Shape change]\label{def:shape}
Let $\{g_\theta\}$ be as in Assumption~\ref{ass:local}, with score $q$ at $\theta = 0$. The change
$\fzero \to g_\theta$ is a \emph{shape change} for a dictionary $\Phi$ if
$\Ezero[\varphi_i(x)\,q(x)] = 0$ for every feature $\varphi_i$ of $\Phi$; equivalently, the
fraction $\kap$ of the Fisher information $I$ captured by $\Phi$, defined in
Theorem~\ref{thm:class}(e), is zero. For
the linear dictionary $\{x\}$ the condition reads $\Ezero[x\,q(x)] = 0$.
\end{definition}

A symmetric variance change at fixed mean, for instance the Gaussian scale family with
$q(x) \propto x^2 - 1$, is a shape change for $\{x\}$ but not for $\{x, x^2\}$. A skewness change
at fixed mean and variance, whose score is orthogonal to $1$, $x$ and $x^2$ under $\fzero$, is a
shape change for $\{x, x^2\}$ as well; the first feature that captures it is cubic, or a
fractional power of the sign-preserving kind. Theorem~\ref{thm:parity} turns these examples into
a structural statement: under symmetry the blindness holds at every order, not only at the order
at which it is first noticed.


\section{Generalized stochastic approximation as a projection}\label{sec:projection}

This section collects the population-level facts behind the calibration step of
Section~\ref{sec:setting}. Proposition~\ref{prop:ku1} identifies the coefficient vector $\vK$ as
the minimiser of a variance-to-separation ratio and as a projection, and records the drift of the
surrogate $\Lam$; Theorem~\ref{thm:info} describes how the information functional grows with the
order. Throughout, $\vphi=(\varphi_1,\dots,\varphi_s)\transp$ is the clipped feature vector of a
dictionary of order $s$, $\uzero,\uone,\Czero,\Cone$ are its means and covariances under
$\fzero,\fone$, and $\matF=\Czero+\Cone$, $\vY=\uone-\uzero$, $\vK=\matF^{-1}\vY$ as in
\eqref{eq:normal-system}. Assumptions~\ref{ass:moments} and~\ref{ass:h1} give
$\matF\succeq\Czero\succ0$, so $\vK$ is well defined.

\subsection{The normal system as a projection}

For a statistic $\psi=k_0+\mathbf{k}\transp\vphi$ with $\mathbf{k}\neq\mathbf{0}$, the \KU{1}
functional of \citet{kunchenko1997moment} is
\begin{equation}\label{eq:ku1}
\mathcal{R}[\psi]
=\frac{\Var_0\psi+\Var_1\psi}{(\Eone\psi-\Ezero\psi)^2}
=\frac{\mathbf{k}\transp(\Czero+\Cone)\mathbf{k}}{(\mathbf{k}\transp\vY)^2},
\end{equation}
the total variance of the statistic relative to the squared separation of its means; it depends
neither on $k_0$ nor on the scale of $\mathbf{k}$. Let $m=\tfrac12(\fzero+\fone)$ be the average
law, $L^2(m)$ the space of functions with finite second moment under $m$, and
\begin{equation}\label{eq:hstar}
h^{*}(x)=\frac{\fone(x)-\fzero(x)}{\fone(x)+\fzero(x)}
=\tanh\bigl(\tfrac12\,\ell(x)\bigr)
\quad\text{wherever }\ell(x)=\ln\frac{\fone(x)}{\fzero(x)}\text{ is finite,}
\end{equation}
the bounded score, $|h^{*}|\le1$. Write $\Pi_s$ for the orthogonal projection in $L^2(m)$ onto
$\mathrm{span}\{1,\varphi_1,\dots,\varphi_s\}$ and $\|\cdot\|_m$ for the norm of $L^2(m)$.

\begin{proposition}[\KU{1} optimality and the normal system]\label{prop:ku1}
Let Assumption~\ref{ass:moments} hold and $\vY\neq\mathbf{0}$.
\begin{enumerate}
\item[(i)] $\mathcal{R}[\psi]\ge2/\Jfunc{s}$ for every $\psi=k_0+\mathbf{k}\transp\vphi$ with
$\mathbf{k}\transp\vY\neq0$, the domain on which \eqref{eq:ku1} is defined, with equality if and
only if $\mathbf{k}=c\vK$ for a scalar $c\neq0$, $\vK=\matF^{-1}\vY$. The stationarity condition of
\eqref{eq:ku1} is
\begin{equation}
\nabla_{\mathbf{k}}\mathcal{R}[\psi]=\mathbf{0}
\;\Longleftrightarrow\;
(\Czero+\Cone)\,\mathbf{k}\propto\vY,
\end{equation}
and \eqref{eq:normal-system} is this condition with the scale fixed by $\matF\vK=\vY$.
\item[(ii)] The surrogate \eqref{eq:surrogate} is a rescaled projection of the bounded score:
\begin{equation}\label{eq:projection}
\Lam=\Bigl(1+\frac{\Jfunc{s}}{2}\Bigr)\Pi_s h^{*},
\qquad
\|\Pi_s h^{*}\|_m^2=\frac{\Jfunc{s}}{2+\Jfunc{s}}.
\end{equation}
For a dictionary that is orthonormal and centred in $L^2(m)$, the entries of $\vY$ are twice the
Fourier coefficients of $h^{*}$, and as the dictionary is completed the partial sums
$\Pi_s h^{*}$ converge to $h^{*}$ in $L^2(m)$ by the Riesz--Fischer theorem
\citep[Ch.~4]{rudin1987real}.
\end{enumerate}
\end{proposition}

Part (ii) is Proposition~1 of \citet{zabolotnii2026spl} and the envelope identity of
\citet{zabolotnii2026mpfd}, in the same metric $\Czero+\Cone$.
The proof in Appendix~\ref{app:projection} is given for completeness and for part (i).

Since $\Ezero\vphi=\uzero$ and $\Eone\vphi=\uone$, \eqref{eq:surrogate} gives
$\Ezero\Lam=\vK\transp\bigl(\uzero-\tfrac12(\uzero+\uone)\bigr)=-\tfrac12\vK\transp\vY
=-\tfrac12\Jfunc{s}$, and likewise $\Eone\Lam=+\tfrac12\Jfunc{s}$. Because $\matF\succ0$ under
Assumption~\ref{ass:moments}, also $\matF^{-1}\succ0$, so
$\Jfunc{s}=\vY\transp\matF^{-1}\vY>0$ whenever $\vY\neq\mathbf{0}$.

Two consequences are used below. First, $\Jfunc{s}$ is the maximum of the Rayleigh quotient
$(\mathbf{k}\transp\vY)^2/(\mathbf{k}\transp\matF\mathbf{k})$ over $\mathbf{k}\neq\mathbf{0}$,
so it depends on the dictionary only through the span of $\{1,\varphi_1,\dots,\varphi_s\}$.
Second, the object that the normal system approximates is $h^{*}$, not the log-likelihood ratio
$\ell$: the two have the same projection direction only in special cases. For
Bernoulli$(0.2)$ against Bernoulli$(0.8)$ with $\varphi_1(x)=x$, the normal system gives
$K_1=0.6/0.32=1.875$, whereas $\ell(x)=(2\ln4)(x-\tfrac12)$ has slope $2\ln4\approx2.77$;
the scale of $\Lam$, not only its direction, is fixed by \eqref{eq:normal-system}, and it is
this scale that Theorem~\ref{thm:class} measures against $\ell$.

\subsection{The information functional}

Let
\begin{equation}\label{eq:tri}
\tri=\int\frac{(\fone-\fzero)^2}{\fone+\fzero}\,\mathrm{d}x\in[0,2]
\end{equation}
be the triangular discrimination between $\fzero$ and $\fone$ (also called the Vincze--Le Cam
distance) \citep{lecam1986asymptotic};
$\tri=2\|h^{*}\|_m^2$, $\tri=0$ if and only if $\fzero=\fone$, and $\tri=2$ if and only if the
two laws are mutually singular.

\begin{theorem}[Information functional]\label{thm:info}
Let Assumption~\ref{ass:moments} hold for every dictionary considered and let
$\Jfunc{s}=\vY\transp\matF^{-1}\vY$ with $\matF=\Czero+\Cone$.
\begin{enumerate}
\item[(a)] (Monotonicity.) If the dictionaries are nested,
$\mathrm{span}\{1,\varphi_1,\dots,\varphi_s\}\subseteq\mathrm{span}\{1,\varphi_1,\dots,\varphi_{s+1}\}$,
then $\Jfunc{s+1}\ge\Jfunc{s}$.
\item[(b)] (Ceiling.) For every dictionary
\begin{equation}\label{eq:ceiling}
\Jfunc{s}=\frac{2\,\|\Pi_s h^{*}\|_m^2}{1-\|\Pi_s h^{*}\|_m^2}
\le\frac{2\tri}{2-\tri},
\end{equation}
the right-hand side read as $+\infty$ when $\tri=2$, with equality if and only if
$h^{*}\in\mathrm{span}\{1,\varphi_1,\dots,\varphi_s\}$; in particular, when $\fzero\neq\fone$,
the single feature $\varphi_1=h^{*}$ is admissible and attains the ceiling. (When $\fzero=\fone$,
$h^{*}=0$ is not a feature under Assumption~\ref{ass:moments}, and every admissible dictionary
attains the ceiling $0$.) If $\tri<2$ and the union of the nested spans is dense
in $L^2(m)$, then $\Jfunc{s}\uparrow2\tri/(2-\tri)$ and $\Lam\to2h^{*}/(2-\tri)$ in $L^2(m)$.
\item[(c)] (The ceiling is not the symmetrised divergence.) Neither
$2\tri/(2-\tri)\le\Dkl(\fone\|\fzero)+\Dkl(\fzero\|\fone)$ nor the reverse inequality holds
in general. For Bernoulli$(0.01)$ against Bernoulli$(0.99)$ the ceiling is
$0.98^2/0.0198\approx48.5$ and the symmetrised divergence is $1.96\ln99\approx9.01$; for
Bernoulli$(0.2)$ against Bernoulli$(0.8)$ the ceiling is $0.36/0.32=1.125$ and the symmetrised
divergence is $1.2\ln4\approx1.664$.
\end{enumerate}
\end{theorem}

The proof is in Appendix~\ref{app:projection}. Part (b) says what a complete dictionary
reconstructs: the bounded score $h^{*}$, whose squared norm is $\tri/2$, and not $\ell$, whose
norm may be infinite. Part (c) records that the limit is a function of $\tri$ alone; the
proportionality $\Jfunc{1}=\delta^2/2=\tfrac12\bigl(\Dkl(\fone\|\fzero)+\Dkl(\fzero\|\fone)\bigr)$
for a Gaussian mean shift of size $\delta$ with $\varphi_1(x)=x$ is a coincidence of that
example, and the two quantities share only their second-order expansion for local alternatives. What the theorem does not say is
that a larger $\Jfunc{s}$ means an earlier alarm at a matched false-alarm level: the delay bound
of Theorem~\ref{thm:class}(c) does carry $\Jfunc{s}$ directly, through
$\Eone\Lam=\tfrac12\Jfunc{s}$, but the tilt root $\thetaz$ of the same theorem rescales the
increment, and what compares the surrogate with $\ell$ is the efficiency $\eff$ of part~(d),
whose local limit is the captured fraction $\kap$ of part~(e).


\section{Structure of the dictionary}\label{sec:structure}

By Proposition~\ref{prop:ku1} the functional $\Jfunc{s}$ depends on the dictionary only through
the span of $\{1,\varphi_1,\dots,\varphi_s\}$, and Theorem~\ref{thm:info} bounds it by the
triangular discrimination; neither says which spans can carry a given change, nor what one
further column is worth. Here a symmetry shared by $\fzero$ and $\fone$ is shown to annihilate a
whole parity of the dictionary at every order (Theorem~\ref{thm:parity}), the exponents of an
unclipped power dictionary to be bounded by the tails (Remark~\ref{rem:admissible}), and one
further column to raise $\Jfunc{s}$ by an explicit ratio (Proposition~\ref{prop:bordered}). The
two-class versions are Propositions~1--2 and the bordered-metric fact of the supplement of
\citet{zabolotnii2026mpfd}, in the same metric $\Czero+\Cone$;
transferred here are the pair (pre-change, post-change), an arbitrary centre $c$ for the origin,
and the clipping of Assumption~\ref{ass:bounded}.

Fix a centre $c\in\real$. A function $\varphi$ is \emph{even about $c$} if
$\varphi(2c-x)=\varphi(x)$ for all $x$ and \emph{odd about $c$} if $\varphi(2c-x)=-\varphi(x)$; a
density $f$ is \emph{symmetric about $c$} if $f(2c-x)=f(x)$ almost everywhere. Thus $x-c$ and
$\operatorname{sgn}(x-c)|x-c|^{a}$ are odd about $c$ while $|x-c|^{a}$ is even: each exponent of the
mixed-parity dictionary $M(A)$ of Section~\ref{sec:dictionaries} contributes one column of each
parity, and the half-line columns of $H(A^{+},A^{-})$, of neither parity, resolve into both.
Throughout, $\chi(y)=\max\bigl(-\phimax,\min(y,\phimax)\bigr)$ is the clipping map and
$\varphi_{\mathrm{lin}}=\chi(x-c)$ the clipped linear column.

\begin{theorem}[Parity]\label{thm:parity}
Let Assumption~\ref{ass:moments} hold and let every feature of the clipped dictionary
$\Phi=(\varphi_1,\dots,\varphi_s)$ be either even or odd about a common centre $c$. Order the
features so that $\vphi=(\vphi_{\mathrm{e}}\transp,\vphi_{\mathrm{o}}\transp)\transp$ collects the
even ones first, and partition $\matF$ and $\vY$ accordingly, with diagonal blocks
$\matF_{\mathrm{ee}},\matF_{\mathrm{oo}}$ and target parts $\vY_{\mathrm{e}},\vY_{\mathrm{o}}$.
\begin{enumerate}
\item[(i)] (Common symmetry.) If $\fzero$ and $\fone$ are both symmetric about $c$, then
$\vY_{\mathrm{o}}=\mathbf{0}$, the cross-parity blocks of $\matF$ vanish, and
$\Jfunc{s}=\vY_{\mathrm{e}}\transp\matF_{\mathrm{ee}}^{-1}\vY_{\mathrm{e}}$, the value of the
functional for the even sub-dictionary alone: every odd feature contributes nothing, whatever the
even ones are. In particular $\Jfunc{s}=0$ at every order $s$ for a dictionary all of whose
features are odd about $c$, the linear dictionary $\{x-c\}$ included.
\item[(ii)] (Mirror pair.) If $\fone(x)=\fzero(2c-x)$, then $\vY_{\mathrm{e}}=\mathbf{0}$, the
cross-parity blocks of $\Czero$ and $\Cone$ cancel in $\matF=\Czero+\Cone$, and
$\Jfunc{s}=\vY_{\mathrm{o}}\transp\matF_{\mathrm{oo}}^{-1}\vY_{\mathrm{o}}$: every even feature
contributes nothing. For a dictionary of even features completed by the clipped linear column
$\varphi_{\mathrm{lin}}$ this leaves
$\Jfunc{s}=2(\Ezero\varphi_{\mathrm{lin}})^{2}/\Var_0\varphi_{\mathrm{lin}}$, which is zero if
and only if $\Ezero\varphi_{\mathrm{lin}}=0$.
\end{enumerate}
\end{theorem}

The proof is in Appendix~\ref{app:structure}. Case~(i) is the one the paper meets: a change of
scale or of kurtosis that leaves both laws symmetric about the same centre is invisible to an odd
dictionary of any length, so raising the order does not repair the blindness. Case~(ii), the
extreme form of a change of skewness at fixed centre, reverses the roles. In both the vanishing
is exact rather than small, so the parity of the dictionary is a design decision taken before any
exponent is chosen.

If in addition the local family of Assumption~\ref{ass:local} preserves the symmetry about $c$,
its score $q$ at $\theta=0$ is even about $c$, so $\Ezero[\varphi_i q]=0$ for every feature odd
about $c$: by Definition~\ref{def:shape} the change is then a shape change for every odd
dictionary, $\kap=0$ and $\Jfunc{s}=0$ at every order, and not only at the order at which the
blindness is first noticed.

\begin{remark}[Admissibility]\label{rem:admissible}
Assumption~\ref{ass:bounded} makes the moment condition of Assumption~\ref{ass:moments}
automatic; without clipping the exponents of a power dictionary are bounded by the tails. Let the
raw features be of the form $|x-c|^{a_i}$ or $\operatorname{sgn}(x-c)|x-c|^{a_i}$ with exponents
$a_i\ge0$, let $a_{\max}$ be the largest of them, the linear column $a=1$ included, and let $\alpha_{\min}=\min(\alpha_0,\alpha_1)$, the tail
index $\alpha_j$ of $f_j$ being the supremum of the orders $q$ with
$\mathbb{E}_j|x-c|^{q}<\infty$. The population pair $(\matF,\vY)$ of the unclipped dictionary
then exists if $2a_{\max}<\alpha_{\min}$ and does not exist if $2a_{\max}>\alpha_{\min}$, the
boundary case depending on the law and being excluded; and if either law is replaced by its
$\eps$-contamination by a law of tail index $\alpha_c<\alpha_{\min}$, the mixture has tail index
$\alpha_c$ for every $\eps\in(0,1)$, so the ceiling on every exponent drops to $\alpha_c/2$
however small: a cubic column needs $\alpha_{\min}>6$, and one contaminant of index three
removes it. This is Proposition~2 of
\citet{zabolotnii2026mpfd} read for the pre- and post-change pair; needing no argument beyond the
monotonicity of $q\mapsto\mathbb{E}|x-c|^{q}$ and the definition of the tail index, it is recorded
here as a remark. Clipping at $\phimax$ replaces the condition rather than satisfies it: with
$\phimax<\infty$ every feature has moments of every order under every law, which is why the
dictionaries of Section~\ref{sec:dictionaries} are clipped. The ceiling $\alpha_{\min}/2$
therefore constrains the unclipped population system and the functional computed from raw
moments; it bears on the design of the exponents, not on the procedure.
\end{remark}

\begin{proposition}[Bordered-metric increment]\label{prop:bordered}
For a dictionary $\Phi$ write $\matF(\Phi)$ and $\vY(\Phi)$ for its metric and its target, and
$\Jfunc{\Phi}=\vY(\Phi)\transp\matF(\Phi)^{-1}\vY(\Phi)$ for the functional \eqref{eq:Jfunc}. Let
$\Phi=(\varphi_1,\dots,\varphi_s)$ and $\Phi\cup\{\psi\}$ both satisfy
Assumption~\ref{ass:moments}, and write the metric and target of the enlarged dictionary in
bordered form,
\[
\matF(\Phi\cup\{\psi\})=\begin{pmatrix}\mathbf{A}&\mathbf{b}\\
\mathbf{b}\transp&d\end{pmatrix},
\qquad
\vY(\Phi\cup\{\psi\})=\begin{pmatrix}\mathbf{y}\\\eta\end{pmatrix},
\]
where $\mathbf{A}=\matF(\Phi)$, $\mathbf{y}=\vY(\Phi)$,
$\mathbf{b}=\mathrm{Cov}_0(\vphi,\psi)+\mathrm{Cov}_1(\vphi,\psi)$,
$d=\Var_0\psi+\Var_1\psi$ and $\eta=\Eone\psi-\Ezero\psi$. Put
$\tilde Y=\eta-\mathbf{b}\transp\mathbf{A}^{-1}\mathbf{y}$, the residual target mass of $\psi$,
and $S=d-\mathbf{b}\transp\mathbf{A}^{-1}\mathbf{b}$, its Schur complement. Then $S>0$ and
\begin{equation}\label{eq:increment}
\Jfunc{\Phi\cup\{\psi\}}-\Jfunc{\Phi}=\frac{\tilde Y^{2}}{S}\ \ge\ 0,
\end{equation}
with equality if and only if $\tilde Y=0$.
\end{proposition}

The proof is in Appendix~\ref{app:structure}. Equation~\eqref{eq:increment} sharpens the
monotonicity of Theorem~\ref{thm:info}(a) into an identity and separates what a column brings
from what it costs: $\tilde Y$ is the separation carried by $\psi$ that the columns present do
not carry, $S$ the dispersion of $\psi$ that they do not explain. A large $\eta$ is not enough on
its own, since it may be reproduced by $\mathbf{b}\transp\mathbf{A}^{-1}\mathbf{y}$, and neither
is an asymmetry of the two laws, which enters only through $\tilde Y$. It also prices a
dictionary fixed in advance: the score detector of Example~\ref{ex:score} is the case $s=1$ with
$\varphi_1=q$, and \eqref{eq:increment} says what a second column adds to it.

\begin{remark}[Choice of the order]\label{rem:order}
Along a nested sequence $\Phi_1\subset\Phi_2\subset\cdots$, \eqref{eq:increment} telescopes to
$\Jfunc{s}=\sum_{r\le s}\tilde Y_r^{2}/S_r$, where $\tilde Y_r$ and $S_r$ are the residual target
mass and the Schur complement of the $r$-th feature against the first $r-1$, and $\Jfunc{0}=0$ is
the functional of the empty dictionary. The terms are non-negative, and since $\fone\ll\fzero$
under Assumption~\ref{ass:moments} forces $\tri<2$, their partial sums are bounded by the finite
ceiling \eqref{eq:ceiling}; hence the terms tend
to zero, and the order is chosen by stopping at the first $s$ at which $\tilde Y_s^{2}/S_s$ falls
below $\eps_J\,\Jfunc{s-1}$. The relative tolerance $\eps_J$ is set by the analyst, from the cost
of one more feature and the accuracy with which the moments entering $\tilde Y_s$ and $S_s$ are
known; no value is recommended here, because what decides is the comparison of the increment with
the sampling error of the moments that produced it, a property of the calibration sample of
Section~\ref{sec:gsa} rather than of the dictionary.
\end{remark}

\paragraph{What this section does not say.} Both statements are about the population functional
$\Jfunc{s}$ and about nothing else. The dictionary with the larger $\Jfunc{s}$ is not thereby the
faster detector: the delay of the procedures of Section~\ref{sec:class} is governed by
$\thetaz\Eone\Lam$ (Theorem~\ref{thm:class}), and two dictionaries with equal $\Jfunc{s}$ need
not have equal slope roots $\thetaz$. Nor does \eqref{eq:increment} say which feature to propose
next: it prices a column already proposed, and the proposal is the analyst's. And symmetry and
moments are exact here: how much of the blindness survives a small departure from symmetry is not
addressed, nor is the estimated system, in which the moments are the calibration-sample ones.

\section{Sequential procedures driven by the increment}\label{sec:procedures}

Throughout, the coefficients are the population ones, $\vK=\matF^{-1}\vY$ of
\eqref{eq:normal-system} with $\matF=\Czero+\Cone$, and $\Lam$ is the increment
\eqref{eq:surrogate}. The object is not one detector but the class $\Tclass$ of
Section~\ref{sec:class}: \eqref{eq:cusum}, \eqref{eq:sr} and \eqref{eq:shiryaev} are three
stopping rules fed the same increment $Z_n=\thetaz\Lam(x_n)$, and the theorem below bounds all
three with the same two population quantities, $\thetaz$ and $\Eone\Lam$.

\begin{theorem}[The class $\Tclass$]\label{thm:class}
Let Assumptions~\ref{ass:moments} and~\ref{ass:bounded} hold, let $\vY\neq\mathbf0$ and
$\Dkl(\fone\|\fzero)<\infty$, and write
$\Lambda^{(s)}_{\max}=\sup_x\Lam(x)\in(0,\infty)$ and $\lambda_p=\log\{1/(1-p)\}$.
\begin{enumerate}
\item[(a)] \emph{(Tilt.)} The equation
\begin{equation}\label{eq:tilt}
\Ezero\exp\bigl(\theta\Lam\bigr)=1
\end{equation}
has the trivial root $\theta=0$ and exactly one root $\thetaz>0$; $\thetaz$ denotes that
positive root throughout. Consequently $\tilde\fzero=e^{\thetaz\Lam}\fzero$ is a probability
density and $e^{Z_n}$ is the likelihood ratio of $\tilde\fzero$ against $\fzero$ at $x_n$.
\item[(b)] \emph{(False alarms, at every threshold.)}
$\Einf T_h\ge e^{h}$ for every $h>0$;
$\Einf T_A\ge A$ for every $A>0$; and, for the Shiryaev rule with geometric prior parameter
$p\in(0,1)$, $\PFA(T_B)\le(1+B)^{-1}$ for every $B>0$. In the last statement, for the increment
$\Lam$ held fixed, $\PFA(T_B)$ is a functional of $\fzero$ and $p$ alone --- $\fone$ enters only
through the coefficients of $\Lam$ and the root $\thetaz$ --- while the bound on it depends on
$B$ alone.
\item[(c)] \emph{(Delay, at every threshold with a positive log-threshold.)} Let
$(\beta,c)=(h,0)$, $(\log A,0)$ and $(\log(B/p),\lambda_p)$ for the three rules, and let the
threshold satisfy $\beta>0$, that is $h>0$, $A>1$ and $B>p$. Then
\[
\ADD(T)\ \le\ \frac{\beta+\thetaz\Lambda^{(s)}_{\max}+c}{\thetaz\Eone\Lam+c},
\qquad \Eone\Lam=\tfrac12\Jfunc{s}>0 ,
\]
and for the Shiryaev rule the delay in Shiryaev's sense,
$\E\bigl[(T_B-\tau+1)^{+}\mid T_B\ge\tau\bigr]$ --- the expectation under the joint law of the
geometric prior and the observations, as in Definition~\ref{def:metrics} --- is bounded by the
same quantity.
\item[(d)] \emph{(One efficiency factor for the class.)} Put
\begin{equation}\label{eq:eff}
\eff=\frac{\thetaz\,\Eone\Lam}{\Dkl(\fone\|\fzero)} .
\end{equation}
Then $0<\eff\le1$, with equality if and only if $\thetaz\Lam=\ell$ $\fzero$-almost surely, the
additive constant being fixed by \eqref{eq:tilt}. Moreover, at thresholds meeting a
false-alarm constraint --- $h=\log\gamma$ and $A=\gamma$ for $\Einf T\ge\gamma$, and
$B=\alpha^{-1}-1$ for $\PFA(T_B)\le\alpha$ --- the delay of each rule exceeds the
corresponding asymptotic optimum, Lorden's \citeyearpar{lorden1971procedures} for the first two
and the Bayesian one \citep[Ch.~6]{tartakovsky2014sequential} for the third, by a factor of at
most $(1+o(1))/\eff$ as $\gamma\to\infty$, respectively $\alpha\to0$; at a finite threshold
the overshoot term of~(c) adds $O(1/\log\gamma)$ to the ratio.
\item[(e)] \emph{(Local limit.)} Let Assumption~\ref{ass:local} hold with $\fone=g_\theta$, and
assume in addition that $\vY(\theta)\neq\mathbf0$ for $\theta>0$ and that
$\Dkl(g_\theta\|\fzero)=\theta^2I/2+o(\theta^2)$. Then, as $\theta\downarrow0$,
$\thetaz\to2$ and
\begin{equation}\label{eq:kappa}
\eff\longrightarrow\kap:=\frac{\vB\transp\Czero^{-1}\vB}{I},
\qquad
\vB_i=\Ezero\bigl[(\vphi-\uzero)_i\,q\bigr],\quad i=1,\dots,s,
\end{equation}
where $0\le\kap\le1$, with $\kap=1$ if and only if $q$ lies in the span of
$(\vphi-\uzero)_1,\dots,(\vphi-\uzero)_s$ $\fzero$-almost surely.
\end{enumerate}
\end{theorem}

The proof is in Appendix~\ref{app:procedures}. The extra expansion assumed in~(e) is not
implied by differentiability in quadratic mean, since a family may place mass $o(\theta^2)$
outside the support of $\fzero$; it holds for the location and scale families used below.

Parts (a) and (b) are elementary consequences of boundedness. Under
Assumption~\ref{ass:bounded} the exponential moment in \eqref{eq:tilt} is finite for every
argument and the overshoot of any of the three statistics is at most one increment, so no normal
or large-$\gamma$ approximation enters: clipping the dictionary is what turns three classical
arguments --- Ville's inequality, the martingale property of $\Rn-n$, and the posterior-odds
identity --- into bounds that hold at every threshold. When $\Lam=\ell$ one has $\thetaz=1$ and
the three bounds are the classical ones; at $s=1$ the exponential bound of (b) sharpens the
Chebyshev threshold brackets attached to the \KU{-}\PE{} criterion of
\citet{kunchenko1997moment}, which control the same tail by a second moment only.

Part (d) replaces any optimality claim, and it speaks about the class and not about a rule: the
constant governing the delay of every member at a given false-alarm level is
$\thetaz\Eone\Lam$, a fraction $\eff$ of the largest possible one, $\Dkl(\fone\|\fzero)$.
Membership of $\ell$ in the span of the dictionary is not sufficient for $\eff=1$, because the
increment is centred at the average law, $\Ezero\Lam+\Eone\Lam=0$ by the drift identities of
Section~\ref{sec:projection}, whereas $\ell$ is normalised by $\Ezero e^{\ell}=1$: for a
one-dimensional span containing $\ell$, $\thetaz\Lam=\ell$ holds if and only if
$\Dkl(\fone\|\fzero)=\Dkl(\fzero\|\fone)$. Unlike $\Jfunc{s}$ in Theorem~\ref{thm:info}, $\eff$
need not be monotone in $s$ at a fixed alternative; $\kap$ is, being a projection norm.

Part (e) identifies the regime the paper is aimed at. The numerator of \eqref{eq:kappa} is the
squared $L^2(\fzero)$ norm of the projection of the score onto the span of the centred features,
so $\kap$ is the fraction of the Fisher information the dictionary carries and the local form of
the information functional, $\Jfunc{s}=\tfrac12\theta^2I\kap+o(\theta^2)$ in the normalisation
$\matF=\Czero+\Cone$; in the terms of Proposition~\ref{prop:ku1}(ii) the two are
squared projection norms of $h^{*}$ and of $q$, agreeing to leading order because
$h^{*}=\tanh(\ell/2)$ and $\ell\approx\theta q$. The same captured-fraction functional governs
estimation, testing and classification built on one truncated likelihood expansion
\citep{godambe1960optimum,zabolotnii2026lsu}; here it is the local ceiling of a detection rule.

\begin{example}[Gaussian mean shift: the classical rules recovered]\label{ex:gaussian}
Let $\fzero=\mathcal N(0,1)$, $\fone=\mathcal N(\mu,1)$ and $\Phi=\{x\}$, and take
$\phimax=\infty$ (clipping changes the identities below by terms of order
$\Pzero(|x|>\phimax)$). With population moments, $\uzero=0$, $\uone=\mu$,
$\Czero=\Cone=1$, $\matF=2$ and $\vK=\mu/2$, so $\Lam$ is half the log-likelihood ratio:
$\Lam(x)=\tfrac12\mu(x-\mu/2)=\tfrac12\ell(x)$. Since
$\Ezero\exp(\theta\Lam)=\exp\{\theta(\theta-2)\mu^2/8\}$, the positive root of \eqref{eq:tilt}
is $\thetaz=2$ and the tilted increment is $\thetaz\Lam=\ell$ exactly;
$\thetaz\Eone\Lam=\mu^2/2=\Dkl(\fone\|\fzero)=\Dkl(\fzero\|\fone)$ and $\eff=1$, in
line with the criterion above, and the three members of $\Tclass$ are the classical rules.
\end{example}

\begin{example}[Scale change: $\ell$ in the span is not enough]\label{ex:scale}
Let $\fzero=\mathcal N(0,1)$, $\fone=\mathcal N(0,1.5^2)$ and $\Phi=\{x,x^2\}$ clipped at
$\phimax$. The log-likelihood ratio $\ell(x)=-\tfrac12(1.5^{-2}-1)x^2-\log1.5$ lies in the span
of $\{1,x,x^2\}$, and the clipped dictionary captures all but a fraction of a percent of the
local Fisher information, yet the efficiency is strictly below one:
$\Jfunc{2}=0.142$ in the normalisation $\matF=\Czero+\Cone$, $\kap=0.995$,
$\thetaz=2.740$ and $\eff=0.884$. The reason is the equality condition of (d): what is needed
is $\thetaz\Lam=\ell$, not $\ell\in\mathrm{span}\,\Phi$. Both laws are symmetric about the
origin, so $Y_1=0$ and the off-diagonal entry of $\matF$ vanishes, whence $K_1=0$ and $\Lam$ is
affine in $x^2$ alone; the one-dimensional criterion above then applies, and $\thetaz\Lam=\ell$
would require $\Dkl(\fone\|\fzero)=\Dkl(\fzero\|\fone)$, that is
$\log(1/1.5)+\tfrac12(1.5^2-1)=\log1.5+(1-1.5^2)/(2\cdot1.5^2)$, which is false. The class
$\Tclass$ therefore pays the factor $1/\eff$ at this alternative, a loss that no threshold and
no calibration sample can remove.
\end{example}

\begin{example}[Symmetric shape change: the dictionary is blind]\label{ex:shape}
Let $g_\theta$ be a shape change in the sense of Definition~\ref{def:shape},
$\Ezero[xq]=0$, and let $\Phi=\{x\}$. Then $\vB=\Ezero[xq]-\uzero\Ezero q=0$, so
$\kappa(\{x\})=0$ and, by (e), the efficiency of every member of $\Tclass$ driven by the linear
increment tends to zero: the delay exceeds the optimal one by a factor growing without bound as
the change becomes small, and since $\eff$ is a population quantity neither the threshold nor
the calibration-sample size changes this. For $g_\theta=\mathcal N(0,e^{2\theta})$,
$q=x^2-1$ and $I=2$; the linear increment vanishes identically ($\vY=\mathbf0$), whereas
$\Phi=\{x,x^2\}$ at $\phimax=\infty$ gives $\vB=(0,2)\transp$, $\Czero=\operatorname{diag}(1,2)$
and $\kap=1$. Clipping acts on the feature value, so $x^2$ is trimmed at $|x|=\phimax^{1/2}$ and
it lowers that figure by terms of order $\Pzero(|x|>\phimax^{1/2})$;
Example~\ref{ex:scale} quotes the clipped value for the same dictionary and the same local
family, and it is below one. When
$\vY=\mathbf0$ exactly the hypotheses of Theorem~\ref{thm:class} are not met: $\Lam\equiv0$,
\eqref{eq:tilt} has no positive root, and none of the three rules ever alarms; such a
configuration is written $\eff=0$ as a convention for ``no signal in the span'', not as a value
the theorem returns, $\thetaz$ is left undefined, and the limit $\thetaz\to2$ in (e) likewise
requires $\kap>0$. Theorem~\ref{thm:parity}(i) is what makes the example structural rather than
particular.
\end{example}

\begin{example}[Dictionaries of order one: the score, and SCUSUM]\label{ex:score}
At $s=1$ the matrix $\matF$ and the vector $\vY$ are scalars and \eqref{eq:surrogate} reads
\[
\thetaz\Lam(x)=\thetaz K_1\Bigl(\varphi_1(x)-\tfrac12\bigl[(\uzero)_1+(\uone)_1\bigr]\Bigr),
\qquad
K_1=\frac{Y_1}{(\Czero)_{11}+(\Cone)_{11}},
\]
an affine function of the single feature whose scale is fixed by \eqref{eq:normal-system} and
whose additive constant is fixed by \eqref{eq:tilt}. Two choices of $\varphi_1$ are worth
naming.

\emph{(i) The score of the local family.} Let Assumption~\ref{ass:local} hold and take
$\varphi_1=q$. If $\phimax\ge\sup_x|q(x)|$, so that the clipping of
Assumption~\ref{ass:bounded} is inactive on this dictionary, part (e) gives $\kap=1$, since $q$
lies in its own span; the efficiency of every member of $\Tclass$ so driven then tends to one as
$\theta\downarrow0$. At a fixed $\theta>0$, however, $\eff<1$ unless $\thetaz\Lam=\ell$, and
Example~\ref{ex:scale} is exactly that gap. If instead $\Pzero(|q|>\phimax)>0$ the clipped
feature spans a different subspace and $\kap$ falls below one by terms of that order, again as
in Example~\ref{ex:scale}.

\emph{(ii) The detector for unnormalised models.} The score-based \CUSUM{} of
\citet{wu2023scusum} is driven by $\lambda\{S_H(x;\fzero)-S_H(x;\fone)\}$, where
$S_H(x;f)=\tfrac12\|\nabla_x\log f(x)\|^2+\Delta_x\log f(x)$ is the Hyv\"arinen score
\citep{hyvarinen2005estimation} and $\lambda>0$ is a free scale. That score is a different
object from the score $q$ of Assumption~\ref{ass:local}: it is a functional of the observation
and of a fixed density, needs no normalising constant, and involves no local family. Taking the
difference $S_H(\cdot;\fzero)-S_H(\cdot;\fone)$ as the single feature $\varphi_1$ makes
\textsc{scusum} the member of $\Tclass$ of order one carrying that feature, and
Theorem~\ref{thm:class} applies to it: the free scale $\lambda$, which that construction leaves
to the user, is supplied here as $\thetaz K_1$, the value that makes $e^{Z_n}$ a likelihood ratio
by part (a), and the centring $\tfrac12(\uzero+\uone)$, which that construction does not
prescribe, is supplied by \eqref{eq:normal-system}. No claim of $\kap=1$ attaches to this
dictionary: what fraction it captures depends on how much of the score of the intended
alternative lies in the span of the Hyv\"arinen-score difference, and that is a property of the
pair, not of the construction.

For either choice, Proposition~\ref{prop:bordered} prices the move to $s>1$: adding a column
raises $\Jfunc{s}$ by exactly the bordered increment $\tilde Y^2/S$, computable from the
moments. That is a gain in the information functional, not in $\eff$, which as noted above need
not be monotone in $s$ at a fixed alternative.
\end{example}

\begin{remark}[Why not a first-order approximation of the divergence]\label{rem:ae}
The delay of every member of $\Tclass$ is governed by $\thetaz\Eone\Lam$ and not by the
captured fraction of the divergence. The two are different functionals: Example~\ref{ex:scale}
has $\kap=0.995$ and $\eff=0.884$, so a dictionary that captures essentially all of the local
Fisher information still loses more than a tenth of the delay constant at a finite alternative,
because $\thetaz\neq1$ there. Choosing a dictionary by a first-order expansion of
$\Dkl(\fone\|\fzero)$ therefore answers a different question from the one the delay asks.
\end{remark}

\begin{remark}[What is not claimed]\label{rem:minimax}
The \CUSUM{} driven by the true log-likelihood ratio is minimax optimal in Lorden's sense
\citep{lorden1971procedures,moustakides1986optimal}. That property belongs to $\ell$, not to
$\Lam$: it is a statement about the statistic the rule is fed, and replacing $\ell$ by
$\thetaz\Lam$ forfeits it unless the two coincide. Theorem~\ref{thm:class} supplies instead
bounds at every threshold and the explicit factor $1/\eff$ against the optimum; no optimality is
asserted for any member of $\Tclass$, nor for the choice of dictionary.
\end{remark}


\section{Estimated coefficients}\label{sec:plugin}

The procedure of Section~\ref{sec:gsa} replaces population moments by estimates: $\huzero$
and $\hCzero$ are the sample mean and the sample covariance of the clipped features on a
calibration sample $x_1,\dots,x_{\Ncal}$ from $\fzero$, and $\hmatF$, $\hvY$,
$\hvK$, $\hat k_0$, $\hat\Lambda$, $\hat J=\hvK\transp\hvY$ and $\hat\thetaz$ are
the sample versions of $\matF$, $\vY$, $\vK$, $k_0$, $\Lam$, $\Jfunc{s}$ and $\thetaz$, the last
being the positive root of $\Ncal^{-1}\sum_{j\le\Ncal}\exp(\theta\hat\Lambda(x_j))=1$. Mode~(M)
of Assumption~\ref{ass:h1} supplies $(\uone,\Cone)$ as a fixed target, so
$(\huone,\hCone)=(\uone,\Cone)$ carries no sampling error; mode~(E) takes
$(\huone,\hCone)$ from an independent labelled sample of size $N_1$ with
$N_1/\Ncal\to\lambda\in(0,\infty]$. For a pair $(\vK,k_0)$ write $e(\vK,k_0)$ for the efficiency
\eqref{eq:eff} of the procedures of Section~\ref{sec:class} driven by $k_0+\vK\transp\vphi$, and
$R(\vK)=(\vK\transp\vY)^2/(\vK\transp\matF\vK)$ for the population Rayleigh quotient, whose
maximum $\Jfunc{s}$ is attained at $\vK=\matF^{-1}\vY$ (Proposition~\ref{prop:ku1}). Finally,
with $\tilde{\vphi}=\vphi-\uzero$ under $\fzero$ and $\tilde{\vphi}_1=\vphi-\uone$ under $\fone$,
collect the mean-zero summands of the two samples in
$\mathbf Z_0=(\tilde{\vphi},\operatorname{vech}(\tilde{\vphi}\tilde{\vphi}\transp-\Czero),
e^{\thetaz\Lam}-1)$ and
$\mathbf Z_1=(\tilde{\vphi}_1,\operatorname{vech}(\tilde{\vphi}_1\tilde{\vphi}_1\transp-\Cone))$,
and put $\mathbf V_0=\operatorname{Cov}_0\mathbf Z_0$, $\mathbf V_1=\operatorname{Cov}_1\mathbf
Z_1$ and $\mathbf V=\mathbf V_0$ in mode~(M),
$\mathbf V=\operatorname{diag}(\mathbf V_0,\lambda^{-1}\mathbf V_1)$ in mode~(E); all three are
finite under Assumption~\ref{ass:bounded}. Write $\boldsymbol\eta=(\uzero,\Czero)$ in mode~(M)
and $\boldsymbol\eta=(\uzero,\Czero,\uone,\Cone)$ in mode~(E) for the moment vector estimated by
the two samples. The remaining coordinate of $\mathbf Z_0$, the sampling term
$e^{\thetaz\Lam}-1$ of the empirical tilt equation, is not a moment and enters part~(d) alone;
so that every derivative below can be paired with $\mathbf V$ as it stands, a Jacobian or
gradient taken with respect to $\boldsymbol\eta$ --- $\mathbf G$, $\mathbf G_{(\vK,k_0)}$ and
$\nabla J$ throughout --- is read as extended by a zero column, respectively a zero entry, on
that coordinate.

\begin{theorem}[Plug-in coefficients]\label{thm:plugin}
Let Assumptions~\ref{ass:moments} and~\ref{ass:bounded} hold, let $\vY\neq\mathbf 0$, and in
mode~(E) let the two samples be independent with $N_1/\Ncal\to\lambda\in(0,\infty]$, as fixed
above; a labelled sample with $N_1=o(\Ncal)$ sets the slower rate $N_1^{-1/2}$ and is not
covered.
\begin{enumerate}
\item[(a)] $\hvK\to\vK$ almost surely and
$\sqrt{\Ncal}(\hvK-\vK)\Rightarrow\mathcal N(0,\boldsymbol\Sigma)$ with
$\boldsymbol\Sigma=\mathbf G\mathbf V\mathbf G\transp$, where $\mathbf G$ is the Jacobian of the
moment map $\boldsymbol\eta\mapsto\matF^{-1}\vY$, given by
\begin{equation}\label{eq:plugin-jac}
\mathrm d\vK=\matF^{-1}\bigl(\mathrm d\vY-\mathrm d\matF\,\vK\bigr),
\qquad
\begin{aligned}
\text{(M)}\quad &\mathrm d\vY=-\mathrm d\uzero, &&\mathrm d\matF=\mathrm d\Czero,\\
\text{(E)}\quad &\mathrm d\vY=\mathrm d\uone-\mathrm d\uzero, &&
\mathrm d\matF=\mathrm d\Czero+\mathrm d\Cone,
\end{aligned}
\end{equation}
Differentiating the centring gives
$\mathrm dk_0=-\tfrac12(\mathrm d\vK)\transp(\uzero+\uone)
-\tfrac12\vK\transp(\mathrm d\uzero+\mathrm d\uone)$, so the pair $(\hvK,\hat k_0)$ is
jointly asymptotically normal with covariance
$\boldsymbol\Sigma_{\vK,k_0}=\mathbf G_{(\vK,k_0)}\mathbf V\mathbf G_{(\vK,k_0)}\transp$,
$\mathbf G_{(\vK,k_0)}$ being the Jacobian of $\boldsymbol\eta\mapsto(\vK,k_0)$. In mode~(M)
both $\mathrm d\uone$ and $\mathrm d\Cone$ vanish and $\boldsymbol\Sigma$ is a functional of the
calibration sample alone; in mode~(E)
$\boldsymbol\Sigma=\boldsymbol\Sigma_0+\lambda^{-1}\boldsymbol\Sigma_1$, the two terms coming from
the two independent samples.
\item[(b)] $R(\hvK)=\Jfunc{s}-(\hvK-\vK)\transp\bigl(\matF-\vY\vY\transp/\Jfunc{s}\bigr)
(\hvK-\vK)\,(1+o_p(1))=\Jfunc{s}-O_p(\Ncal^{-1})$; under Assumption~\ref{ass:local} with
$\vB\neq\mathbf 0$ the same second-order expansion holds for
the quotient $(\mathbf k\transp\vB)^2/(I\,\mathbf k\transp\Czero\mathbf k)$, whose maximum over
$\mathbf k\neq\mathbf 0$ is the captured fraction $\kap$ of \eqref{eq:kappa}, about its maximiser
$\Czero^{-1}\vB$. The efficiency, when $\Dkl(\fone\|\fzero)<\infty$, satisfies
$e(\hvK,\hat k_0)-e(\vK,k_0)=\nabla e\transp(\hvK-\vK,\,\hat k_0-k_0)\transp
+O_p(\Ncal^{-1})$ with
\begin{equation}\label{eq:plugin-eff}
\nabla e=\frac{\thetaz}{\Dkl(\fone\|\fzero)}
\left[\begin{pmatrix}\uone\\1\end{pmatrix}
-\frac{\Eone\Lam}{\tilde{\E}_0\Lam}\begin{pmatrix}\tilde{\E}_0\vphi\\1\end{pmatrix}\right],
\qquad \tilde{\E}_0h=\Ezero[h\,e^{\thetaz\Lam}]
\end{equation}
the expectation under the exponentially tilted law. The gradient vanishes if and only if the
tilted law reproduces $\uone$ and $\Eone\Lam$; this holds when $\thetaz\Lam=\ell$, the case of
equality in \eqref{eq:eff}, but not in general, so the plug-in efficiency loses
$O_p(\Ncal^{-1/2})$, with asymptotic variance $\nabla e\transp\boldsymbol\Sigma_{\vK,k_0}\nabla e$
computed as in~(a).
\item[(c)] The relative standard error of the $s$-th diagonal entry of $\hCzero$ --- the
sample variance of $\varphi_s$, the feature of highest exponent in the dictionaries of
Section~\ref{sec:dictionaries} --- is $\sqrt{(\gamma_4(\varphi_s)+2)/\Ncal}+O(\Ncal^{-1})$,
$\gamma_4(\varphi_s)$ being the excess kurtosis of the clipped feature under $\fzero$; it is
finite for every law by Assumption~\ref{ass:bounded}. The leading term vanishes only when
$\gamma_4(\varphi_s)=-2$, that is, when $\varphi_s$ takes two values with equal probability,
and the error is then of order $\Ncal^{-1}$.
\item[(d)] $\sqrt{\Ncal}(\hat J-\Jfunc{s})\Rightarrow\mathcal N(0,\nabla J\transp\mathbf V\nabla
J)$ and $\sqrt{\Ncal}(\hat\thetaz-\thetaz)\Rightarrow\mathcal N(0,\nabla\thetaz\transp\mathbf
V\nabla\thetaz)$, where $\nabla J$ is the gradient of $\boldsymbol\eta\mapsto\Jfunc{s}$ obtained
from $\mathrm dJ=2\vK\transp\mathrm d\vY-\vK\transp\mathrm d\matF\,\vK$ --- that is, $2\vK$ with
respect to $\vY$ and $-\vK\vK\transp$ with respect to $\matF$ --- composed with
\eqref{eq:plugin-jac}, and, by implicit differentiation of $\Ezero e^{\theta\Lam}=1$ in
$(\theta,\vK,k_0)$,
\begin{equation}\label{eq:plugin-an}
\frac{\partial\thetaz}{\partial(\vK,k_0)}
=-\frac{\thetaz}{\tilde{\E}_0[\Lam]}\begin{pmatrix}\tilde{\E}_0\vphi\\ 1\end{pmatrix}\transp,
\qquad
\nabla\thetaz=\Bigl(\frac{\partial\thetaz}{\partial(\vK,k_0)}\,\mathbf G_{(\vK,k_0)}\Bigr)\transp
-\frac{1}{\tilde{\E}_0[\Lam]}\,\mathbf e,
\end{equation}
in which $\partial\thetaz/\partial(\vK,k_0)$ is a row of length $s+1$, $\mathbf G_{(\vK,k_0)}$ is
the Jacobian of~(a), and $\mathbf e$ is the unit vector of the tilt coordinate of $\mathbf Z_0$;
that second term is what makes $\nabla\thetaz$ more than the chain rule through $(\vK,k_0)$.
Writing
$\mathbf g_i=\mathbf G_{(\vK,k_0)}^{(i)\transp}(\tilde{\E}_0\vphi,1)\transp$ for the blocks of
$\mathbf G_{(\vK,k_0)}$ acting on the calibration sample ($i=0$) and on the labelled sample
($i=1$), and $\mathbf V_0^{-}$ for the covariance of $\mathbf Z_0$ with its last coordinate
removed, the limiting variance of $\hat\thetaz$ is
\begin{equation*}
\frac{1}{(\tilde{\E}_0\Lam)^2}
\Bigl[\Var_0 e^{\thetaz\Lam}
+2\thetaz\,\mathbf g_0\transp\operatorname{Cov}_0\bigl(\mathbf Z_0^{-},e^{\thetaz\Lam}\bigr)
+\thetaz^{2}\,\mathbf g_0\transp\mathbf V_0^{-}\mathbf g_0\Bigr]
+\frac{\thetaz^{2}}{\lambda\,(\tilde{\E}_0\Lam)^2}\,
\mathbf g_1\transp\mathbf V_1\mathbf g_1,
\end{equation*}
the last term absent in mode~(M).
\end{enumerate}
\end{theorem}

The proof is in Appendix~\ref{app:plugin}. Part~(b) separates two statements. The direction of
the plug-in coefficients is optimal to first order --- the captured information loses only
$O_p(\Ncal^{-1})$, because $\vK$ maximises the Rayleigh quotient, and the same expansion governs
every quotient of that shape, the one behind $\kap$ included --- while the efficiency of the
resulting procedure is not stationary at $\vK$: the
normal system optimises a quadratic proxy, not $\eff$, so a first-order term survives unless
$\Lam$ is exact. In the local regime of Theorem~\ref{thm:class}(e) that term disappears, since
$e(\vK,k_0)\to\kap$, which is stationary at the population direction; it is an effect of a fixed,
non-local $\thetaz$.

Part~(d) is not part~(a) applied twice. The estimator $\hat\thetaz$ reads the calibration sample
once through $(\hvK,\hat k_0)$ and once through the empirical expectation in the tilt
equation, and the two readings are strongly dependent: for the $t_5$ location pair of
Section~\ref{sec:numerics} with the quadratic polynomial dictionary the two contributions to
\eqref{eq:plugin-an} are of the same size with correlation $-0.94$, so the limiting variance of
$\hat\thetaz$ is about a seventh of either of them, and either reading alone would misstate it by
a factor of that order. What a finite calibration sample can report of that variance is a
separate matter, and the cancellation is what makes it hard: on the same pair at
$\Ncal=N_1=2000$, the standard error of $\hat J$ estimated inside a single pair of samples reads
$0.01261$ against the asymptotic $0.01262$, whereas for $\hat\thetaz$ the same estimate reads
$0.0724$ against $0.0852$, an estimate of a small difference of two nearly equal terms. That is a
measurement at one size, not a size to use.

\begin{remark}[Calibration size]\label{rem:ncal}
Part~(c) turns the choice of $\Ncal$ into a condition on one number: for a target relative
standard error $r$ of the leading variance entry, $\Ncal\ge(\gamma_4(\varphi_s)+2)/r^2$. The
target is the analyst's; Table~\ref{tab:ncal} evaluates the bound at two points of that axis,
and recommends neither. Clipping acts on the feature value, so $\varphi_3$ is trimmed at
$|x|=\phimax^{1/3}$, which is why the $s=3$ row is not larger than the $s=2$ row; without it the
$s\ge2$ entries need $\Ezero x^{4s}$ and are infinite for $t_5$, leaving the condition empty.
\end{remark}
\begin{table}[t]
\centering
\caption{Theorem~\ref{thm:plugin}(c) for the polynomial dictionary clipped at $\phimax=10$: excess kurtosis $\gamma_4(\varphi_s)$ of the clipped feature under $\fzero$, and the smallest $\Ncal$ at which the relative standard error of $(\hat{\mathbf{C}}_0)_{ss}$ is at most $r$. The two $r$ are points on the axis of the bound, not recommendations.}
\label{tab:ncal}
\begin{tabular}{llrrr}
\toprule
Law & $s$ & $\gamma_4(\varphi_s)$ & $\Ncal$ ($r=0.10$) & $\Ncal$ ($r=0.05$) \\
\midrule
Student $t_5$ & 1 & 3.45 & 546 & 2181 \\
Student $t_5$ & 2 & 5.71 & 771 & 3084 \\
Student $t_5$ & 3 & 3.00 & 501 & 2002 \\
Laplace & 2 & 3.55 & 555 & 2220 \\
\bottomrule
\end{tabular}
\end{table}


\section{Numerical illustration}\label{sec:numerics}

\begin{table}[p]
\centering
\scriptsize
\setlength{\tabcolsep}{4pt}
\caption{Efficiency of the surrogate driving a \CUSUM{} or a Shiryaev--Roberts (\SR{}) statistic. Population quantities by quadrature with the clipped dictionaries ($\phimax = 10$) and $\matF = \Czero+\Cone$: the information functional $\Jfunc{s}$, the tilt root $\thetaz$ of \eqref{eq:tilt}, the efficiency $\eff$ of \eqref{eq:eff} and the captured Fisher-information fraction $\kap$ of \eqref{eq:kappa} for the local family through $\fzero$. The last two columns are the measured ratio of detection delays of the oracle arm (true log-likelihood ratio) to the surrogate arm, for the \CUSUM{} ($g_t = \max(0, g_{t-1}+\text{increment})$) and for the \SR{} statistic ($R_t = (1+R_{t-1})\exp(\text{increment})$, $R_0 = 0$), both at $\ARL = 1000$ on the same streams: thresholds $h$ (\CUSUM{}) and $\log A$ (\SR{}) from the same 2000 in-control runs of length 12000, delays from the same 2000 runs started at the change with $g_0 = 0$ resp.\ $R_0 = 0$, common random numbers across the oracle and surrogate arms within each statistic, seed 42. All three dictionaries have $\varphi_1(x)=x$, so their $s=1$ rows coincide and are printed once (\texttt{all}). The local family behind $\kap$ is the location family for the first three pairs, the scale family for the lognormal and the Gaussian variance pair, and the Pearson~III skewness family for the last. A dash marks a configuration with $\vY = \mathbf{0}$ (no signal in the span, so $\Lam \equiv 0$) or, in either of the last two columns, a predicted delay beyond the simulation budget. The identity between either ratio column and $\eff$ holds only as $\ARL \to \infty$, so at a finite target neither column is a direct check of $\eff$.}
\label{tab:efficiency}
\begin{tabular}{llcrrrrrr}
\toprule
$(\fzero,\fone)$ & $\Phi$ & $s$ & $\Jfunc{s}$ & $\thetaz$ & $\eff$ & $\kap$ & $\ADD_{\ell}/\ADD_{\Lam}$ \CUSUM{} & \SR{} \\
\midrule
$\mathcal{N}(0,1)\to\mathcal{N}(0.5,1)$ & \texttt{all} & 1 & 0.1250 & 2.000 & 1.000 & 1.000 & $1.000 \pm 0.000$ & $1.000 \pm 0.000$ \\
\addlinespace[2pt]
 & \texttt{poly} & 2 & 0.1250 & 2.000 & 1.000 & 1.000 & $1.000 \pm 0.001$ & $1.001 \pm 0.000$ \\
 &  & 3 & 0.1251 & 1.996 & 0.999 & 1.000 & $1.002 \pm 0.003$ & $1.001 \pm 0.001$ \\
\addlinespace[2pt]
 & \texttt{frac} & 2 & 0.1251 & 1.996 & 0.999 & 1.000 & $1.001 \pm 0.003$ & $1.003 \pm 0.002$ \\
 &  & 3 & 0.1252 & 1.994 & 0.998 & 1.000 & $1.002 \pm 0.003$ & $1.002 \pm 0.002$ \\
\addlinespace[2pt]
 & \texttt{log} & 2 & 0.1250 & 2.000 & 1.000 & 1.000 & $1.001 \pm 0.001$ & $1.003 \pm 0.001$ \\
 &  & 3 & 0.1252 & 1.994 & 0.998 & 1.000 & $1.005 \pm 0.003$ & $1.006 \pm 0.002$ \\
\midrule
$t_5\to t_5+0.5$ & \texttt{all} & 1 & 0.0755 & 1.916 & 0.780 & 0.805 & $0.852 \pm 0.008$ & $0.839 \pm 0.006$ \\
\addlinespace[2pt]
 & \texttt{poly} & 2 & 0.0757 & 1.892 & 0.773 & 0.805 & $0.843 \pm 0.008$ & $0.834 \pm 0.006$ \\
 &  & 3 & 0.0771 & 1.829 & 0.761 & 0.825 & $0.837 \pm 0.008$ & $0.825 \pm 0.006$ \\
\addlinespace[2pt]
 & \texttt{frac} & 2 & 0.0900 & 1.954 & 0.949 & 0.968 & $0.963 \pm 0.005$ & $0.942 \pm 0.005$ \\
 &  & 3 & 0.0927 & 1.942 & 0.971 & 0.995 & $0.989 \pm 0.004$ & $0.974 \pm 0.004$ \\
\addlinespace[2pt]
 & \texttt{log} & 2 & 0.0764 & 1.898 & 0.783 & 0.805 & $0.848 \pm 0.008$ & $0.843 \pm 0.006$ \\
 &  & 3 & 0.0924 & 1.945 & 0.970 & 0.983 & $0.986 \pm 0.004$ & $0.973 \pm 0.004$ \\
\midrule
$\mathrm{Lap}(0,1)\to\mathrm{Lap}(0.5,1)$ & \texttt{all} & 1 & 0.0625 & 1.942 & 0.570 & 0.500 & $0.674 \pm 0.009$ & $0.658 \pm 0.007$ \\
\addlinespace[2pt]
 & \texttt{poly} & 2 & 0.0627 & 1.911 & 0.563 & 0.500 & $0.660 \pm 0.009$ & $0.654 \pm 0.007$ \\
 &  & 3 & 0.0704 & 1.797 & 0.594 & 0.583 & $0.686 \pm 0.009$ & $0.689 \pm 0.007$ \\
\addlinespace[2pt]
 & \texttt{frac} & 2 & 0.0957 & 1.954 & 0.878 & 0.922 & $0.920 \pm 0.007$ & $0.916 \pm 0.006$ \\
 &  & 3 & 0.0959 & 1.960 & 0.882 & 0.994 & $0.924 \pm 0.007$ & $0.921 \pm 0.006$ \\
\addlinespace[2pt]
 & \texttt{log} & 2 & 0.0706 & 1.862 & 0.617 & 0.500 & $0.689 \pm 0.009$ & $0.692 \pm 0.007$ \\
 &  & 3 & 0.1033 & 1.927 & 0.935 & 0.831 & $0.956 \pm 0.006$ & $0.960 \pm 0.005$ \\
\midrule
$\mathrm{LN}(0,1)\to 1.3\cdot\mathrm{LN}(0,1)$ & \texttt{all} & 1 & 0.0252 & 2.016 & 0.738 & 0.718 & $0.812 \pm 0.010$ & $0.829 \pm 0.008$ \\
\addlinespace[2pt]
 & \texttt{poly} & 2 & 0.0275 & 2.040 & 0.816 & 0.771 & $0.884 \pm 0.009$ & $0.889 \pm 0.007$ \\
 &  & 3 & 0.0287 & 2.037 & 0.850 & 0.806 & $0.902 \pm 0.009$ & $0.908 \pm 0.006$ \\
\addlinespace[2pt]
 & \texttt{frac} & 2 & 0.0313 & 1.875 & 0.853 & 0.933 & $0.897 \pm 0.008$ & $0.932 \pm 0.006$ \\
 &  & 3 & 0.0344 & 1.998 & 0.997 & 0.997 & $0.998 \pm 0.003$ & $0.997 \pm 0.002$ \\
\addlinespace[2pt]
 & \texttt{log} & 2 & 0.0344 & 2.000 & 1.000 & 1.000 & $1.000 \pm 0.000$ & $1.000 \pm 0.000$ \\
 &  & 3 & 0.0344 & 2.000 & 1.000 & 1.000 & $1.000 \pm 0.001$ & $1.001 \pm 0.000$ \\
\midrule
$\mathcal{N}(0,1)\to\mathcal{N}(0,1.5^2)$ & \texttt{all} & 1 & 0.0000 & -- & 0.000 & 0.000 & -- & -- \\
\addlinespace[2pt]
 & \texttt{poly} & 2 & 0.1416 & 2.740 & 0.884 & 0.995 & $0.947 \pm 0.007$ & $0.994 \pm 0.006$ \\
 &  & 3 & 0.1416 & 2.740 & 0.884 & 0.995 & $0.947 \pm 0.007$ & $0.994 \pm 0.006$ \\
\addlinespace[2pt]
 & \texttt{frac} & 2 & 0.0000 & -- & 0.000 & 0.000 & -- & -- \\
 &  & 3 & 0.0000 & -- & 0.000 & 0.000 & -- & -- \\
\addlinespace[2pt]
 & \texttt{log} & 2 & 0.0666 & 2.430 & 0.369 & 0.405 & $0.503 \pm 0.008$ & $0.475 \pm 0.007$ \\
 &  & 3 & 0.0666 & 2.430 & 0.369 & 0.405 & $0.503 \pm 0.008$ & $0.475 \pm 0.007$ \\
\midrule
P-III $\gamma_3\!:0\to1.5$ & \texttt{all} & 1 & 0.0000 & -- & 0.000 & 0.000 & -- & -- \\
\addlinespace[2pt]
 & \texttt{poly} & 2 & 0.0015 & 1.841 & 0.007 & 0.000 & $0.052 \pm 0.001$ & $0.046 \pm 0.001$ \\
 &  & 3 & 0.1001 & 1.747 & 0.407 & 0.554 & $0.550 \pm 0.007$ & $0.524 \pm 0.006$ \\
\addlinespace[2pt]
 & \texttt{frac} & 2 & 0.0818 & 2.301 & 0.438 & 0.529 & $0.608 \pm 0.007$ & $0.580 \pm 0.006$ \\
 &  & 3 & 0.0901 & 2.540 & 0.532 & 0.764 & $0.690 \pm 0.007$ & $0.668 \pm 0.007$ \\
\addlinespace[2pt]
 & \texttt{log} & 2 & 0.0005 & 1.856 & 0.002 & 0.000 & $0.038 \pm 0.001$ & $0.033 \pm 0.000$ \\
 &  & 3 & 0.0935 & 2.447 & 0.532 & 0.713 & $0.685 \pm 0.007$ & $0.666 \pm 0.006$ \\
\bottomrule
\end{tabular}
\end{table}

The examples of Section~\ref{sec:procedures} are analytic. This section adds one Monte Carlo
study, whose purpose is to place the two quantities of Theorem~\ref{thm:class}(d) --- the
population factor $\eff$ and the delay ratio it governs --- beside each other over a range of
pairs and dictionaries wide enough that agreement would not be an accident.
Table~\ref{tab:efficiency} reports, for six pre- and post-change pairs and the three
dictionaries at $s\in\{1,2,3\}$, the population values of $\Jfunc{s}$, $\thetaz$, $\eff$
and $\kap$ by quadrature in the normalisation $\matF=\Czero+\Cone$ of
\eqref{eq:normal-system}, beside two measured ratios of the oracle's detection delay to the
surrogate's: one for the \CUSUM{} \eqref{eq:cusum}, one for the Shiryaev--Roberts rule
\eqref{eq:sr}. The caption fixes every calibration quantity of the experiment --- how many
in-control runs, of what length, which threshold, which seed --- and they are not repeated
here.

Counts below run over all $54$ configurations, the shared $s=1$ row expanded back into its
three dictionaries; the table prints $42$.
The population columns behave as Theorem~\ref{thm:class}(d) requires. First, $\eff\le1$ in
every configuration;
for the Gaussian mean shift at $s=1$, $\thetaz=2.000$ and $\eff=1.000$, which is
Example~\ref{ex:gaussian}, where the tilted surrogate $\thetaz\Lam$ is the log-likelihood
ratio itself.
Second, $\Jfunc{s}$ rises at every one of the $36$ nested steps, as Theorem~\ref{thm:info}(a)
requires, while $\eff$ falls at $7$ of them, by at most $1.6\%$ (Student-$t$, polynomial,
$s=2\to3$). There is no contradiction: the added feature carries information, but it also
lengthens the right tail of $\Lam$ and so lowers the tilt root $\thetaz$ of \eqref{eq:tilt},
and it is the product $\thetaz\Eone\Lam$, not $\Jfunc{s}$ alone, that \eqref{eq:eff} divides
by the divergence.
Third, on the Gaussian variance change and on the Pearson~III skewness change the linear
dictionary has $\vY=\mathbf0$ and $\Lam\equiv0$, so $\eff=0.000$ --- by the parity mechanism
of Theorem~\ref{thm:parity} in the first case, and because the pair holds the mean and the
variance fixed in the second --- whereas the polynomial dictionary reaches $0.884$ at $s=2$
on the variance change and $0.407$ at $s=3$ on the skewness change. What the dictionary can
resolve, and not the size of the change, decides whether there is anything to detect.

Both ratio columns are finite-$\gamma$ measurements and neither is a direct check of $\eff$,
since Theorem~\ref{thm:class}(d) equates the two only as $\gamma\to\infty$. At the target
$\ARL=1000$ of this experiment each arm carries an $O(1)$ overshoot worth a fraction of order
$1/\log\gamma$ of its threshold, and that overshoot lengthens the shorter oracle delay
proportionally more, so the measured \CUSUM{} ratio is expected to lie above $\eff$. It does
so in $42$ of the $46$ rows that have a measured delay, in $35$ of them by more than two
standard errors of the paired ratio; the three Gaussian mean-shift rows at $s=1$, in which the
tilted surrogate is the oracle, differ from it only by the rounding of the threshold onto its
grid, a gap of $-0.0002$, and are counted in neither figure; the fourth exception, the
lognormal pair with the logarithmic dictionary at $s=3$, lies $0.0001$ below $\eff$, within
one standard error. The signed gaps run from $-0.0002$ to $0.170$ and widen as $\eff$ falls.
Raising the target closes the gap. A separate sweep, with $400$ in-control runs of length
$240\,000$ for the thresholds and $1000$ post-change runs for the delays, takes the Laplace
pair through $\ARL=10^3\to5\cdot10^3\to2\cdot10^4$: the ratio falls from $0.687$ to $0.648$
to $0.631$ with the polynomial dictionary at $s=1$, where $\eff=0.570$, and from $0.952$ to
$0.952$ to $0.940$ with the logarithmic dictionary at $s=3$, where $\eff=0.935$; the Gaussian
row in which the tilted surrogate is the oracle stays within $0.001$ of $1.000$ throughout.

The two ratio columns together are the empirical side of Theorem~\ref{thm:class}(d) being a
statement about the class $\Tclass$ of Section~\ref{sec:class} rather than about one rule.
The \CUSUM{} and the Shiryaev--Roberts statistic are fed the same tilted increment on the same
streams and are given the same efficiency factor by the theorem; the two measured ratios
agree to within two standard errors of their paired difference in $29$ of the $46$ rows in
which both are measured, the Shiryaev--Roberts ratio lying above the \CUSUM{} ratio in $20$
of them and below in the rest, and the largest absolute difference anywhere in the table is
$0.047$, with a standard error of $0.010$, on the Gaussian variance pair with the polynomial
dictionary at $s=2$, where the Shiryaev--Roberts ratio is $0.994$ against the \CUSUM{} ratio
$0.947$. The residual
disagreement is of the size the theorem leaves open. Part~(c) separates the two rules by
their threshold term alone, and at a finite $\ARL$ the two recursions accumulate different
overshoots; what the class statement fixes is the common factor $\thetaz\Eone\Lam$ in the
denominator, not the $O(1)$ terms above it. Across the whole range of $\eff$ that the table
spans, neither column departs from the population factor in a way that depends on which of
the two rules carries the increment.


\section{Open questions}\label{sec:open}

\emph{Dependence.} Everything in Sections~\ref{sec:procedures} and~\ref{sec:class} rests on
one structural fact: under independence the tilted increment of \eqref{eq:tilt} makes
$\exp\{\sum_{i\le n}\thetaz\Lam(x_i)\}$ a martingale under $\Pinf$, and every bound of
Theorem~\ref{thm:class}(b) is a stopping-time identity for it. For a dependent stream the
same route needs a martingale built from conditional increments, so the tilt root would have
to be defined from the conditional law of $x_n$ given its past rather than from $\fzero$;
whether the dictionary moments still determine that root is not answered here.

\emph{Choosing the dictionary.} The bordered increment of Proposition~\ref{prop:bordered}
prices one further column exactly, which turns the choice of $\Phi$ into an optimisation over
candidate features under a budget on $s$. This paper takes $\Phi$ as given; the selection
problem, including how far the greedy rule that the bordered increment suggests can fall
below the best subset of a given size, is untouched.

\emph{The finite-sample side of parity.} Theorem~\ref{thm:parity} is a population statement:
under the stated symmetry the relevant components of $\vY$ vanish exactly, and
Example~\ref{ex:shape} is the corresponding blind detector. With coefficients read off a
calibration sample the estimate of $\vY$ is not zero, so what one runs is a small increment
rather than none, and whether the resulting statistic is merely uninformative or actively
misleading at a fixed threshold is a question about the joint law of the estimate and the run
length that the asymptotics here do not reach.


%
%

\section*{Disclosure of interest}

The author reports there are no competing interests to declare.

\section*{Funding}

No funding was obtained for the work reported in this article.

\section*{Declaration of generative AI and AI-assisted technologies in the manuscript preparation process}

During the preparation of this work the author used Claude Code (Anthropic; Claude Fable~5.1 and
Claude Opus~5 models) to implement and organise the code of the Monte-Carlo study, to assist with
the \LaTeX{} typesetting, and for language editing. Drafts of the proofs were additionally read
for gaps by GPT-6 Astra (OpenAI) and Gemini~3.1 Pro (Google) in the role of adversarial referees;
every objection so raised was checked by the author against the proof before any change was made,
and no statement or proof was generated by these tools. After using these tools the author
reviewed and edited the material as needed and takes full responsibility for the content of the
publication.

\section*{Data availability statement}

This article studies specified laws and uses no measured data, so no data set was collected or
analysed for it. The code that produces Table~\ref{tab:efficiency}, together with the experiment
scripts and the result files each table names, is available at
\url{https://github.com/SZabolotnii/KuYuPe-Change_Point-code-supplement}. The Monte-Carlo study of
Section~\ref{sec:numerics} is run from a fixed seed and is exactly reproducible.

\bibliography{references}
\appendix
\section{Proofs}\label{app:proofs}

\subsection{Proofs for Section~\ref{sec:projection}}\label{app:projection}

\subsubsection{Proof of Proposition~\ref{prop:ku1}}\label{app:ku1}

Part (ii) follows the \KU{1} ratio route of \citet{zabolotnii2026spl} and the envelope identity
of \citet{zabolotnii2026mpfd}, both stated in the same metric $\Czero+\Cone$; part (i), the
scale-invariant lower bound and its reduction to $\matF\vK=\vY$, is proved here for
completeness.

\emph{Variation of the \KU{1} ratio.}
For $\psi=k_0+\mathbf{k}\transp\vphi$, $\Var_i\psi=\mathbf{k}\transp\mathbf{C}_i\mathbf{k}$
and $\Eone\psi-\Ezero\psi=\mathbf{k}\transp\vY$, which gives the second form in
\eqref{eq:ku1}; $k_0$ cancels and $\mathcal{R}$ is homogeneous of degree $0$ in $\mathbf{k}$.
With $\mathbf{S}=\Czero+\Cone=2\matF$, differentiation of the quotient on
$\{\mathbf{k}\transp\vY\neq0\}$ gives
\[
\nabla_{\mathbf{k}}\mathcal{R}
=\frac{2}{(\mathbf{k}\transp\vY)^3}
\Bigl[(\mathbf{k}\transp\vY)\,\mathbf{S}\mathbf{k}-(\mathbf{k}\transp\mathbf{S}\mathbf{k})\,\vY\Bigr],
\]
so $\nabla_{\mathbf{k}}\mathcal{R}=\mathbf{0}$ if and only if
$\mathbf{S}\mathbf{k}=\lambda\vY$ with
$\lambda=\mathbf{k}\transp\mathbf{S}\mathbf{k}/\mathbf{k}\transp\vY$, which is the
stationarity condition of Proposition~\ref{prop:ku1}(i). Since $\mathbf{S}=2\matF\succ0$ under
Assumption~\ref{ass:moments}, every stationary point is a multiple of $\matF^{-1}\vY$, and
the choice $\lambda=2$ is \eqref{eq:normal-system}.

\emph{Minimum value.}
In the inner product $\langle\mathbf{a},\mathbf{b}\rangle_{\matF}=\mathbf{a}\transp\matF\mathbf{b}$
the Cauchy--Schwarz inequality reads
\[
(\mathbf{k}\transp\vY)^2=\langle\mathbf{k},\matF^{-1}\vY\rangle_{\matF}^2
\le(\mathbf{k}\transp\matF\mathbf{k})(\vY\transp\matF^{-1}\vY),
\]
with equality if and only if $\mathbf{k}\propto\matF^{-1}\vY$. Hence
$\mathcal{R}[\psi]=2\,\mathbf{k}\transp\matF\mathbf{k}/(\mathbf{k}\transp\vY)^2\ge2/\Jfunc{s}$,
which is (i), and equivalently
\[
\Jfunc{s}=\max_{\mathbf{k}\neq\mathbf{0}}
\frac{(\mathbf{k}\transp\vY)^2}{\mathbf{k}\transp\matF\mathbf{k}}.
\]
Replacing $\mathbf{k}$ by $c\mathbf{k}$ leaves $\mathcal{R}$ unchanged, so only the direction of
$\vK$ matters for detection; the scale fixed by $\matF\vK=\vY$ is what makes $\Lam$ the surrogate
used throughout.

\emph{Projection identity.}
Let $\bar{\mathbf{u}}=\tfrac12(\uzero+\uone)=\E_m\vphi$ and
$\tilde{\vphi}=\vphi-\bar{\mathbf{u}}$, so that $\Lam=\vK\transp\tilde{\vphi}$ by
\eqref{eq:surrogate}. Under $m$,
\[
\mathrm{Cov}_m(\vphi)
=\tfrac12\bigl(\Ezero\vphi\vphi\transp+\Eone\vphi\vphi\transp\bigr)
-\bar{\mathbf{u}}\bar{\mathbf{u}}\transp
=\tfrac12\matF+\tfrac14\vY\vY\transp,
\qquad
\E_m[\tilde{\vphi}\,h^{*}]=\tfrac12\int\tilde{\vphi}\,(\fone-\fzero)\,\mathrm{d}x=\tfrac12\vY,
\]
because $h^{*}m=\tfrac12(\fone-\fzero)$ and $\E_m h^{*}=0$. The projection of $h^{*}$ onto
$\mathrm{span}\{1,\vphi\}$ is $\Pi_s h^{*}=\mathbf{a}\transp\tilde{\vphi}$ with
$\mathbf{a}=\mathrm{Cov}_m(\vphi)^{-1}\E_m[\tilde{\vphi}h^{*}]
=\bigl(\matF+\tfrac12\vY\vY\transp\bigr)^{-1}\vY$; by the Sherman--Morrison
formula, with $J=\Jfunc{s}$,
\[
\bigl(\matF+\tfrac12\vY\vY\transp\bigr)^{-1}\vY
=\frac{\matF^{-1}\vY}{1+\tfrac12\vY\transp\matF^{-1}\vY}
=\frac{2\vK}{2+J},
\qquad\text{so}\qquad
\mathbf{a}=\frac{2\vK}{2+J},
\]
which is the first identity in \eqref{eq:projection}. The second follows from
\[
\|\Pi_s h^{*}\|_m^2=\mathbf{a}\transp\mathrm{Cov}_m(\vphi)\,\mathbf{a}
=\frac{4}{(2+J)^2}\Bigl(\tfrac12\vK\transp\matF\vK+\tfrac14(\vK\transp\vY)^2\Bigr)
=\frac{4}{(2+J)^2}\cdot\frac{J(2+J)}{4}
=\frac{J}{2+J}.
\]
For a dictionary orthonormal and centred in $L^2(m)$, $\E_m[\varphi_i h^{*}]$ is the Fourier
coefficient $\langle\varphi_i,h^{*}\rangle_m$ and $Y_i=2\langle\varphi_i,h^{*}\rangle_m$; the
Riesz--Fischer theorem gives $\|\Pi_s h^{*}-h^{*}\|_m\to0$ when the system is complete.
\qed

\subsubsection{Proof of Theorem~\ref{thm:info}}\label{app:info}

Part (a) is the nested-subspace projection argument; part (b) is the Cauchy--Schwarz envelope of
the MPFD supplement \citep{zabolotnii2026mpfd}, in the same metric; the
Bernoulli counterexamples of part (c) are new here.

Put $\rho_s=\|\Pi_s h^{*}\|_m^2$. By \eqref{eq:projection}, $\rho_s=J/(2+J)$, that is
$\Jfunc{s}=2\rho_s/(1-\rho_s)$, which is the equality in \eqref{eq:ceiling}; the map
$t\mapsto2t/(1-t)$ is continuous and strictly increasing on $[0,1)$, and $\rho_s<1$ because
$\Jfunc{s}<\infty$ under Assumption~\ref{ass:moments}.

(a) Nested spans give $\Pi_s=\Pi_s\Pi_{s+1}$, hence $\rho_s\le\rho_{s+1}$ and
$\Jfunc{s}\le\Jfunc{s+1}$.

(b) $\rho_s\le\|h^{*}\|_m^2$, with equality if and only if $h^{*}$ lies in the closed,
finite-dimensional span, and
\[
\|h^{*}\|_m^2
=\int\frac{(\fone-\fzero)^2}{(\fone+\fzero)^2}\cdot\frac{\fone+\fzero}{2}\,\mathrm{d}x
=\frac{\tri}{2}.
\]
Monotonicity of $t\mapsto2t/(1-t)$ gives
$\Jfunc{s}\le2(\tri/2)/(1-\tri/2)=2\tri/(2-\tri)$ with the same equality condition, and the
span of $\{1,h^{*}\}$ contains $h^{*}$; the single feature $h^{*}$ is admissible when
$\fzero\neq\fone$, since $h^{*}$ constant $\fzero$-almost surely would force $\fone=k\fzero$ on the
support of $\fzero$, hence $k=1$ by $\fone\ll\fzero$ and normalisation. If the union of the nested
spans is dense in $L^2(m)$,
the projections onto an increasing sequence of closed subspaces with dense union converge
strongly to the identity, so $\rho_s\uparrow\tri/2<1$ and
$\Jfunc{s}\uparrow2\tri/(2-\tri)$ by continuity; then
$1+\Jfunc{s}/2\to2/(2-\tri)$ and \eqref{eq:projection} gives
$\Lam\to2h^{*}/(2-\tri)$ in $L^2(m)$.

(c) On the two-point space $\{0,1\}$ with $\fzero=(1-p,p)$ and $\fone=(1-r,r)$, the feature
$\varphi_1(x)=x$ spans all functions modulo constants, so $\Jfunc{1}$ equals the ceiling:
$\vY=r-p$, $\matF=p(1-p)+r(1-r)$, $\Jfunc{1}=(r-p)^2/\matF$, while
$\Dkl(\fone\|\fzero)+\Dkl(\fzero\|\fone)=(r-p)\bigl(\ln\tfrac{r}{1-r}-\ln\tfrac{p}{1-p}\bigr)$.
For $(p,r)=(0.01,0.99)$ these are $0.9604/0.0198\approx48.5$ and $1.96\ln99\approx9.01$, the
ceiling above the divergence; for $(p,r)=(0.2,0.8)$ they are $0.36/0.32=1.125$ and
$1.2\ln4\approx1.664$, the ceiling below it.
\qed


\subsection{Proofs for Section~\ref{sec:structure}}\label{app:structure}

\subsubsection{Proof of Theorem~\ref{thm:parity}}

The route is the proof of Proposition~1 of \citet{zabolotnii2026mpfd}, written about the centre
$c$ instead of the origin and for the clipped dictionary of Assumption~\ref{ass:bounded}; what
the transfer needs is that clipping does not disturb the parity split, and that is settled first.
Every moment below is a moment of a bounded function, finite by Assumption~\ref{ass:bounded}.

\emph{Clipping preserves parity about $c$.} The clip $\chi$ of
Section~\ref{sec:structure} is odd, $\chi(-y)=-\chi(y)$, so the clipped
form $\varphi=\chi\circ\tilde\varphi$ of a raw feature $\tilde\varphi$ odd about $c$ satisfies
$\varphi(2c-x)=\chi\bigl(-\tilde\varphi(x)\bigr)=-\varphi(x)$, and one even about $c$ satisfies
$\varphi(2c-x)=\chi(\tilde\varphi(x))=\varphi(x)$. Each clipped column has the parity of the raw
column it clips, and a product of two of them is even about $c$ when the parities agree and odd
when they differ. Finally, $x\mapsto2c-x$ is an involution with unit Jacobian, so
$\int g(x)f(x)\,\mathrm{d}x=\int g(2c-u)f(2c-u)\,\mathrm{d}u$ for every bounded $g$ and every
density $f$; both parts of the theorem are this reflection identity applied to a suitable $g$.

\emph{(i) Common symmetry.} Let $f_j(2c-x)=f_j(x)$, $j=0,1$. The identity with $f=f_j$ gives
$\mathbb{E}_j[g(x)]=\mathbb{E}_j[g(2c-x)]$. With $g=\varphi$ odd about $c$ it reads
$\mathbb{E}_j\varphi=-\mathbb{E}_j\varphi$, so $\mathbb{E}_j\varphi=0$ for $j=0,1$: the odd parts
of $\uzero$ and $\uone$ vanish and $\vY_{\mathrm{o}}=\mathbf{0}$. With $g=\varphi_u\varphi_v$,
$\varphi_u$ even and $\varphi_v$ odd about $c$, the product is odd about $c$, so
$\mathbb{E}_j[\varphi_u\varphi_v]=0$; since also $\mathbb{E}_j\varphi_v=0$,
$\mathrm{Cov}_j(\varphi_u,\varphi_v)=\mathbb{E}_j[\varphi_u\varphi_v]
-\mathbb{E}_j\varphi_u\,\mathbb{E}_j\varphi_v=0$ for $j=0,1$. Both cross-parity blocks vanish,
hence so does that of $\matF=\Czero+\Cone$, and
$\matF=\mathrm{diag}(\matF_{\mathrm{ee}},\matF_{\mathrm{oo}})$ with
$\vY=(\vY_{\mathrm{e}}\transp,\mathbf{0}\transp)\transp$. Under Assumption~\ref{ass:moments}
$\matF\succ0$, so both principal blocks are positive definite, and the block-diagonal inverse
gives $\Jfunc{s}=\vY_{\mathrm{e}}\transp\matF_{\mathrm{ee}}^{-1}\vY_{\mathrm{e}}$ whatever
$\matF_{\mathrm{oo}}$ is. If every feature is odd about $c$ the even block is empty,
$\vY=\mathbf{0}$ and $\Jfunc{s}=0$ at every order; applied to the clipped column $x-c$ alone this
is $\Jfunc{1}=0$ for the linear dictionary.

\emph{(ii) Mirror pair.} Let $\fone(x)=\fzero(2c-x)$. The identity with $f=\fzero$ gives
$\Eone[g(x)]=\int g(x)\fzero(2c-x)\,\mathrm{d}x=\Ezero[g(2c-x)]$ for every bounded $g$. Hence
$\Eone\varphi_u=\Ezero\varphi_u$ for $\varphi_u$ even about $c$, so $\vY_{\mathrm{e}}=\mathbf{0}$,
and $\Eone\varphi_v=-\Ezero\varphi_v$ for $\varphi_v$ odd about $c$. For a cross-parity pair the
product is odd about $c$, so $\Eone[\varphi_u\varphi_v]=-\Ezero[\varphi_u\varphi_v]$ and
\begin{equation*}
\mathrm{Cov}_1(\varphi_u,\varphi_v)
=-\Ezero[\varphi_u\varphi_v]-\Ezero\varphi_u\bigl(-\Ezero\varphi_v\bigr)
=-\mathrm{Cov}_0(\varphi_u,\varphi_v).
\end{equation*}
The two cross-parity blocks are of opposite sign and cancel in
$\matF=\Czero+\Cone$, neither vanishing on its own. As in (i), $\matF$ is block diagonal in parity, now with the target on the odd block, so
$\Jfunc{s}=\vY_{\mathrm{o}}\transp\matF_{\mathrm{oo}}^{-1}\vY_{\mathrm{o}}$ and every even
feature contributes nothing. If the even features are completed by the clipped linear column
$\varphi_{\mathrm{lin}}$ alone, the odd block is one-dimensional, with
$Y_{\mathrm{o}}=-2\,\Ezero\varphi_{\mathrm{lin}}$ and, since
$\Eone[\varphi_{\mathrm{lin}}^{2}]=\Ezero[\varphi_{\mathrm{lin}}^{2}]$ and
$(\Eone\varphi_{\mathrm{lin}})^{2}=(\Ezero\varphi_{\mathrm{lin}})^{2}$,
$\matF_{\mathrm{oo}}=\Var_0\varphi_{\mathrm{lin}}+\Var_1\varphi_{\mathrm{lin}}
=2\Var_0\varphi_{\mathrm{lin}}$; hence
$\Jfunc{s}=2(\Ezero\varphi_{\mathrm{lin}})^{2}/\Var_0\varphi_{\mathrm{lin}}$, zero if and only
if $\Ezero\varphi_{\mathrm{lin}}=0$.
\qed

\subsubsection{Proof of Proposition~\ref{prop:bordered}}

This is the redundancy criterion of \citet{breusch1999redundancy} written for the metric $\matF$;
the identity is proved for completeness. New here are only its reading as the price of one
dictionary column and the fact that $\matF=\Czero+\Cone$, being a sum of covariance
matrices, is symmetric, which is what the squared form requires.

Write $\matF_{\cup},\vY_{\cup}$ for the bordered metric and target of the statement. By
Assumption~\ref{ass:moments} applied to $\Phi\cup\{\psi\}$, $\matF_{\cup}\succ0$; hence
$\mathbf{A}\succ0$ is invertible and $S=d-\mathbf{b}\transp\mathbf{A}^{-1}\mathbf{b}>0$ as the
Schur complement of a positive definite matrix. Put
$\vK_{\cup}=\bigl(\mathbf{A}^{-1}\mathbf{y}-\tfrac{\tilde Y}{S}\mathbf{A}^{-1}\mathbf{b},\,
\tfrac{\tilde Y}{S}\bigr)$, written as a column. Its upper block under $\matF_{\cup}$ is
$\mathbf{y}-\tfrac{\tilde Y}{S}\mathbf{b}+\tfrac{\tilde Y}{S}\mathbf{b}=\mathbf{y}$ and its lower
block is
$\mathbf{b}\transp\mathbf{A}^{-1}\mathbf{y}
+\tfrac{\tilde Y}{S}(d-\mathbf{b}\transp\mathbf{A}^{-1}\mathbf{b})
=\mathbf{b}\transp\mathbf{A}^{-1}\mathbf{y}+\tilde Y=\eta$ by the definition of $\tilde Y$. So
$\matF_{\cup}\vK_{\cup}=\vY_{\cup}$, and as $\matF_{\cup}$ is invertible $\vK_{\cup}$ is the
solution of \eqref{eq:normal-system} for the enlarged dictionary. By \eqref{eq:Jfunc}, applied to
$\Phi\cup\{\psi\}$ and to $\Phi$,
\begin{equation*}
\Jfunc{\Phi\cup\{\psi\}}=\vY_{\cup}\transp\vK_{\cup}
=\mathbf{y}\transp\mathbf{A}^{-1}\mathbf{y}
+\frac{\tilde Y}{S}\Bigl(\eta-\mathbf{y}\transp\mathbf{A}^{-1}\mathbf{b}\Bigr)
=\Jfunc{\Phi}+\frac{\tilde Y}{S}\Bigl(\eta-\mathbf{y}\transp\mathbf{A}^{-1}\mathbf{b}\Bigr),
\end{equation*}
which is the identity with no symmetry hypothesis on $\mathbf{A}$. Here $\mathbf{A}$ is a
principal submatrix of $\matF=\Czero+\Cone$, a sum of covariance matrices, so
$\mathbf{A}$ and $\mathbf{A}^{-1}$ are symmetric and
$\mathbf{y}\transp\mathbf{A}^{-1}\mathbf{b}=\mathbf{b}\transp\mathbf{A}^{-1}\mathbf{y}
=\eta-\tilde Y$; the bracket is then $\tilde Y$ and the increment is $\tilde Y^{2}/S$, which is
\eqref{eq:increment}. Since $S>0$ it is non-negative and vanishes if and only if $\tilde Y=0$.
\qed


\subsection{Proofs for Section~\ref{sec:procedures}}\label{app:procedures}

\subsubsection{Proof of Theorem~\ref{thm:class}}

Throughout, $Z_i=\thetaz\Lam(x_i)$ is the increment of Section~\ref{sec:class}, and
$\lambda_p=\log\{1/(1-p)\}\ge0$. Every entry of $\vphi$ and of the centre
$\tfrac12(\uzero+\uone)$ is bounded by $\phimax$ in absolute value, so
Assumption~\ref{ass:bounded} gives $|\Lam|\le2\phimax\|\vK\|_1<\infty$ and every exponential
moment used below is finite;
this single fact is what replaces the truncation arguments that an unbounded increment
would require, and it is used at each optional-stopping step.

\paragraph{(a) The positive root and the tilted law.}
The construction is the exponential change of measure of \citet[Ch.~3]{siegmund1985sequential},
applied here to a statistic that is not a log-likelihood ratio. Put
$\psi(t)=\log\Ezero e^{t\Lam}$. Boundedness of $\Lam$ makes $\psi$ finite and real-analytic on
$\real$, with $\psi(0)=0$ and, by the drift identities of Section~\ref{sec:projection},
$\psi'(0)=\Ezero\Lam=-\tfrac12\Jfunc{s}<0$, the inequality because
Assumption~\ref{ass:moments} gives $\matF\succ0$ and $\vY\neq\mathbf0$. If
$\Pzero(\Lam>0)=0$ then $\Lam\le0$ $\fzero$-almost surely, hence $\fone$-almost surely because
$\fone\ll\fzero$, so $\Eone\Lam\le0$, contradicting $\Eone\Lam=\tfrac12\Jfunc{s}>0$. Therefore
$\pi=\Pzero(\Lam>\eps)>0$ for some $\eps>0$ and $e^{\psi(t)}\ge\pi e^{t\eps}\to\infty$ as
$t\to\infty$; in particular $\Lam$ is not $\fzero$-almost surely constant, so
$\psi''(t)=\Var_{t}\Lam>0$ under the law tilted at $t$ and $\psi$ is strictly convex. Such a
$\psi$ is positive on $(-\infty,0)$, negative on some $(0,\thetaz)$ and positive on
$(\thetaz,\infty)$ for exactly one $\thetaz>0$, so $\theta=0$ and $\theta=\thetaz$ are the only
roots of \eqref{eq:tilt}. Since $\psi(\thetaz)=0$, $\tilde\fzero=e^{\thetaz\Lam}\fzero$ is a
probability density with $\tilde\fzero/\fzero=e^{\thetaz\Lam}$.

\paragraph{(b) \CUSUM{}.}
This is Ville's inequality \citep{ville1939etude} inserted into Theorem~2 of
\citet{lorden1971procedures}. Write $S_n=\sum_{i\le n}Z_i$ and
$N_h=\inf\{n\ge1:S_n\ge h\}$ for the one-sided crossing time started from zero. The recursion
\eqref{eq:cusum} gives $W_n=S_n-\min_{0\le k\le n}S_k$, so
$T_h=\min_{k\ge1}\bigl(N_h^{(k)}+k-1\bigr)$, where $N_h^{(k)}$ is $N_h$ applied to
$x_k,x_{k+1},\dots$; by \citet[Theorem~2]{lorden1971procedures}, if
$\Pzero(N_h<\infty)\le\alpha$ then $\Einf T_h\ge1/\alpha$ and $\ADD(T_h)\le\Eone N_h$. Under
$\Pinf$ the variables $x_i$ are i.i.d.\ $\fzero$, so $M_n=e^{S_n}$ is a nonnegative martingale
with $\Einf M_n=1$ by independence and \eqref{eq:tilt}, and Ville's inequality gives
$\Pinf(N_h<\infty)=\Pinf\bigl(\sup_nM_n\ge e^{h}\bigr)\le e^{-h}$; Lorden's theorem with
$\alpha=e^{-h}$ yields $\Einf T_h\ge e^{h}$. (Had the recursion been driven by $\Lam$ rather
than by $\thetaz\Lam$, the same computation would give $e^{\thetaz h}$.)

\paragraph{(b) Shiryaev--Roberts.}
The mechanism is the martingale property of the Shiryaev--Roberts statistic under $\Pinf$, in
the form used by \citet{pollak1987average}. Under $\Pinf$ the observation $x_n$ is independent
of $\mathcal F_{n-1}=\sigma(x_1,\dots,x_{n-1})$ and has density $\fzero$, so
$\Einf\bigl[e^{Z_n}\mid\mathcal F_{n-1}\bigr]=\Ezero e^{\thetaz\Lam}=1$ by part (a). The
recursion \eqref{eq:sr} therefore gives
$\Einf[\Rn\mid\mathcal F_{n-1}]=(1+R_{n-1})\Einf[e^{Z_n}\mid\mathcal F_{n-1}]=1+R_{n-1}$,
that is, $\Rn-n$ is a $\Pinf$-martingale with $R_0-0=0$. It is integrable at every $n$:
$0\le\Rn\le n\,e^{n\thetaz\Lambda^{(s)}_{\max}}<\infty$ by Assumption~\ref{ass:bounded} and the
recursion, so $\Einf|\Rn-n|<\infty$. For a fixed $n$ the time $T_A\wedge n$ is a bounded
stopping time, and optional stopping applied to the integrable martingale $\Rn-n$ gives
$\Einf[R_{T_A\wedge n}]=\Einf[T_A\wedge n]$.
On $\{T_A\le n\}$ one has $R_{T_A\wedge n}=R_{T_A}\ge A$ by the definition of $T_A$, and on
$\{T_A>n\}$ one has $R_n\ge0$; hence
$\Einf[T_A\wedge n]\ge A\,\Pinf(T_A\le n)$. Letting $n\to\infty$, the left side increases to
$\Einf T_A$ by monotone convergence and $\Pinf(T_A\le n)\uparrow\Pinf(T_A<\infty)$, so
$\Einf T_A\ge A\,\Pinf(T_A<\infty)$. If $\Pinf(T_A<\infty)=1$ this is the claim; if
$\Pinf(T_A<\infty)<1$ then $T_A=\infty$ with positive probability and $\Einf T_A=\infty\ge A$.
No uniform integrability is invoked: the only inputs are the boundedness of $\Lam$, which makes
$\Rn$ integrable, and monotone convergence.

\paragraph{(b) Shiryaev.}
The route is the Bayesian one of \citet{shiryaev1963optimum} in the form of
\citet[Ch.~6]{tartakovsky2014sequential}: the false-alarm probability is read off the posterior
odds of $\{\tau\le n\}$.

\emph{Step 1: for a fixed increment, the probability does not see $\fone$.} Throughout this
step the increment $\Lam$, hence the stopping rule $T_B$, is held fixed; $\fone$ enters $T_B$
only through the population coefficients and the root $\thetaz$, which are not recomputed when
the post-change density is replaced below. Under the geometric prior
$\Pr(\tau=k)=p(1-p)^{k-1}$, $k\ge1$, the change time is independent of the observation
mechanism, and $\{T_B<\tau\}\cap\{\tau=k\}=\{T_B\le k-1\}\cap\{\tau=k\}$ with
$\{T_B\le k-1\}\in\mathcal F_{k-1}$. Conditionally on $\tau=k$ the variables $x_1,\dots,x_{k-1}$
are i.i.d.\ $\fzero$ whatever the post-change density is, so
$\PFA(T_B)=\sum_{k\ge1}p(1-p)^{k-1}\Pinf(T_B\le k-1)$,
a functional of $\fzero$ and $p$ only. It is therefore unchanged when $\fone$ is replaced by any
other post-change density, and we may evaluate it in the model $\tilde{\Pr}$ whose post-change
density is the tilted $\tilde\fzero$ of part (a), the prior being the same.

\emph{Step 2: $\Sn$ is the posterior odds.} In $\tilde{\Pr}$ the likelihood ratio of one
post-change observation is $\tilde\fzero(x)/\fzero(x)=e^{\thetaz\Lam(x)}=e^{Z}$, so Bayes' rule
with the geometric prior gives, for the posterior odds
$\Omega_n=\tilde{\Pr}(\tau\le n\mid\mathcal F_n)/\tilde{\Pr}(\tau>n\mid\mathcal F_n)$,
\[
\Omega_n=\frac{\sum_{k=1}^{n}p(1-p)^{k-1}\prod_{j=k}^{n}e^{Z_j}}{(1-p)^{n}} ,
\qquad \Omega_0=0 .
\]
Splitting off the $k=n$ summand, which contributes $p\,e^{Z_n}/(1-p)$, and factoring
$e^{Z_n}/(1-p)$ out of the remaining $n-1$ summands gives
$\Omega_n=(p+\Omega_{n-1})e^{Z_n}/(1-p)$ --- which is exactly the recursion
\eqref{eq:shiryaev}. Hence $\Sn=\Omega_n$ for every $n$, and
$\tilde{\Pr}(\tau>n\mid\mathcal F_n)=1/(1+\Sn)$.

\emph{Step 3: evaluation at the alarm.} Because $\{T_B=n\}\in\mathcal F_n$,
\[
\PFA(T_B)=\tilde{\Pr}(T_B<\tau)
=\sum_{n\ge1}\tilde{\Pr}(\tau>n,\ T_B=n)
=\sum_{n\ge1}\tilde{\E}\Bigl[\frac{\mathbf 1\{T_B=n\}}{1+\Sn}\Bigr]
\le\frac{1}{1+B},
\]
since $S_{T_B}\ge B$ on $\{T_B<\infty\}$ and the summands vanish on $\{T_B=\infty\}$; the
interchange needs no integrability condition, the integrand being bounded by one, and Step 1
transports the bound back to the actual model. The bound is the classical one because the
threshold of \eqref{eq:shiryaev} is applied to the posterior odds itself; a statistic
proportional to $\Omega_n$ with a different constant would carry that constant into the bound.

\paragraph{(c) Delay.}
The mechanism is Wald's identity in the form of \citet[Theorem~2]{lorden1971procedures}, with
the overshoot controlled by Assumption~\ref{ass:bounded}. Each of the three statistics dominates
a random walk restarted at any index, which is what makes one computation serve all three.

\emph{A pathwise lower bound.} Fix $1\le k\le n$. From \eqref{eq:cusum}, $W_n\ge W_{n-1}+Z_n$
and $W_n\ge0$, so induction downwards from $k$ gives $W_n\ge\sum_{j=k}^{n}Z_j$. From
\eqref{eq:sr}, $\log\Rn=\log(1+R_{n-1})+Z_n\ge\log R_{n-1}+Z_n$ and $\log\Rn\ge Z_n$ because
$1+R_{n-1}\ge1$, so $\log\Rn\ge\sum_{j=k}^{n}Z_j$. From \eqref{eq:shiryaev}, dropping $S_{n-1}$
gives $\Sn\ge p\,e^{Z_n+\lambda_p}$ and dropping $p$ gives $\Sn\ge S_{n-1}e^{Z_n+\lambda_p}$, so
$\log\Sn\ge\log p+\sum_{j=k}^{n}(Z_j+\lambda_p)$: starting the statistic at $S_0=0$ costs the
constant $\log(1/p)$, which is the only place where the initial value enters. Each bound holds
for every $k$ and does not involve the state at time $k-1$.

\emph{A renewal at $\tau$.} Fix a change time $\tau$ and a history $x_1,\dots,x_{\tau-1}$; under
$\mathbb P_\tau$ the variables $x_\tau,x_{\tau+1},\dots$ are i.i.d.\ $\fone$ and independent of
it. Write $\xi_j=Z_j+c$ with $c=0$ for the first two rules and $c=\lambda_p$ for the third, and
let $N=\inf\bigl\{m\ge1:\sum_{j=\tau}^{\tau+m-1}\xi_j\ge\beta\bigr\}$ with $\beta=h$, $\log A$,
$\log(B/p)$ respectively, each positive by hypothesis. The pathwise bound at $k=\tau$ gives
$T\le\tau-1+N$ on every path, hence $(T-\tau+1)^{+}\le N$ whatever the history is, and
$\ADD(T)\le\Eone N$ by Definition~\ref{def:metrics}.

\emph{Wald's identity.} The increments $\xi_j$, $j\ge\tau$, are i.i.d.\ with
$\Eone\xi=\thetaz\Eone\Lam+c>0$ and $|\xi|\le2\thetaz\phimax\|\vK\|_1+c<\infty$. Put
$G_m=\sum_{j=\tau}^{\tau+m-1}\xi_j$ and $G_0=0$. For each $n\ge1$ the time $N\wedge n$ is a
bounded stopping time of that sequence, so Wald's identity gives
$\Eone[G_{N\wedge n}]=(\thetaz\Eone\Lam+c)\Eone[N\wedge n]$. On $\{N\le n\}$ one has
$G_{N-1}<\beta$ --- for $N=1$ this is $G_0=0<\beta$, which is where positivity of the
log-threshold is used --- so
$G_{N\wedge n}=G_N<\beta+\thetaz\Lambda^{(s)}_{\max}+c$; on $\{N>n\}$, $G_n<\beta$. Hence
$\Eone[N\wedge n]\le(\beta+\thetaz\Lambda^{(s)}_{\max}+c)/(\thetaz\Eone\Lam+c)$ for every $n$,
and monotone convergence gives the same bound for $\Eone N$; in particular $N<\infty$
$\Pone$-almost surely. Substituting $(\beta,c)=(h,0)$, $(\log A,0)$, $(\log(B/p),\lambda_p)$ gives the
three bounds of (c), with $\Eone\Lam=\tfrac12\Jfunc{s}$ by the drift identities of
Section~\ref{sec:projection}.

\emph{The Shiryaev-sense delay.} Given $\tau=k$, the event $\{T_B\ge k\}=\{T_B>k-1\}$ lies in
$\mathcal F_{k-1}$, so conditioning on $\mathcal F_{k-1}$ inside each term of
$\E[(T_B-\tau+1)^{+}\mathbf 1\{T_B\ge\tau\}]
=\sum_{k\ge1}\Pr(\tau=k)\E_k[\mathbf 1\{T_B>k-1\}(T_B-k+1)^{+}]$
and bounding the inner conditional expectation by $\ADD(T_B)$, as
Definition~\ref{def:metrics} permits, gives
$\E[(T_B-\tau+1)^{+}\mathbf 1\{T_B\ge\tau\}]\le\ADD(T_B)\Pr(T_B\ge\tau)$. Dividing by
$\Pr(T_B\ge\tau)$ gives the last claim of (c).

\paragraph{(d) Efficiency.}
The upper half is again the change of measure of \citet[Ch.~3]{siegmund1985sequential}, in its
variational form; the lower halves are the bound of \citet{lorden1971procedures} for the
run-length constraint and its Bayesian counterpart in \citet[Ch.~6]{tartakovsky2014sequential}.

For a bounded measurable $g$ put $\tilde f=e^{g}\fzero/\Ezero e^{g}$, a density with
$\fone\ll\tilde f$. Then
$\Dkl(\fone\|\fzero)-\Eone g+\log\Ezero e^{g}=\Eone\log(\fone/\tilde f)=\Dkl(\fone\|\tilde f)\ge0$,
with equality if and only if $\fone=\tilde f$ almost surely; this is the Donsker--Varadhan
inequality \citep[Cor.~4.15]{boucheron2013concentration}. Taking $g=\thetaz\Lam$, the logarithm
vanishes by \eqref{eq:tilt}, so $\Dkl(\fone\|\fzero)\ge\thetaz\Eone\Lam>0$, which is
$0<\eff\le1$, with equality if and only if $\fone=e^{\thetaz\Lam}\fzero$ almost surely, that is
$\thetaz\Lam=\ell$ $\fzero$-almost surely; since also $\Ezero e^{\ell}=1$, no additive constant
is free. For the one-dimensional criterion quoted in the text, let
$\Lam=K_1(\varphi_1-\bar u)$ with $\bar u=\tfrac12\{(\uzero)_1+(\uone)_1\}$ and
$\ell=c_0+a\varphi_1$. Then $\thetaz\Lam=\ell$ forces $\thetaz K_1=a$ and $c_0+a\bar u=0$, that
is $\Ezero\ell+\Eone\ell=0$; conversely, if $c_0+a\bar u=0$ then $\ell=(a/K_1)\Lam$ with
$a/K_1>0$, because $aY_1=\Eone\ell-\Ezero\ell>0$ and $K_1Y_1=\Jfunc{1}>0$, and
$\Ezero e^{\ell}=1$ identifies $\thetaz=a/K_1$. Since $\Eone\ell=\Dkl(\fone\|\fzero)$ and
$\Ezero\ell=-\Dkl(\fzero\|\fone)$, the condition reads
$\Dkl(\fone\|\fzero)=\Dkl(\fzero\|\fone)$.

For the asymptotic statements, take $h=\log\gamma$ and $A=\gamma$: by (b) both give
$\Einf T\ge\gamma$, and by (c) $\ADD(T)\le\log\gamma/(\thetaz\Eone\Lam)+O(1)$, the constant
being $\Lambda^{(s)}_{\max}/\Eone\Lam$. Every stopping time with $\Einf T\ge\gamma$ has
$\ADD(T)\ge(1+o(1))\log\gamma/\Dkl(\fone\|\fzero)$ \citep{lorden1971procedures}, so the ratio of
the two is at most $(1+o(1))/\eff$. For the Shiryaev rule take $B=\alpha^{-1}-1$, so that
$\PFA(T_B)\le\alpha$ by (b); then (c) gives
$\ADD(T_B)\le\log(1/\alpha)/(\thetaz\Eone\Lam+\lambda_p)+O(1)$ as $\alpha\downarrow0$ at fixed
$p$, the $O(1)$ absorbing $\log(1/p)$ and the overshoot, while the Bayesian lower bound is
$(1+o(1))\log(1/\alpha)/\{\Dkl(\fone\|\fzero)+\lambda_p\}$
\citep[Ch.~6]{tartakovsky2014sequential}. Their ratio,
$\{\Dkl(\fone\|\fzero)+\lambda_p\}/(\thetaz\Eone\Lam+\lambda_p)$, is at most $1/\eff$ because
$\lambda_p\ge0$, $\thetaz\Eone\Lam\le\Dkl(\fone\|\fzero)$ and
$(v+c)/(w+c)\le v/w$ whenever $0<w\le v$ and $c\ge0$.

\paragraph{(e) Local limit.}
The limit functional $\vB\transp\Czero^{-1}\vB/I$ is the fraction of the Fisher information
captured by a finite set of estimating functions in the sense of \citet{godambe1960optimum}, and
the functional that governs estimation, testing and classification from one truncated likelihood
expansion in \citet{zabolotnii2026lsu}. Write $\tilde\vphi=\vphi-\uzero$ and $\E_\theta$ for
expectation under $g_\theta$.

\emph{Step 1.} For measurable $u$ with $\sup|u|\le1$,
$\E_\theta u-\Ezero u=\theta\,\Ezero[uq]+r_\theta(u)$ with $\sup_u|r_\theta(u)|=o(\theta)$: with
$a=\tfrac12\theta q\sqrt{\fzero}$ and $b=\sqrt{g_\theta}-\sqrt{\fzero}-a$,
Assumption~\ref{ass:local} gives $\|b\|_2=o(\theta)$ and $\|a\|_2=\tfrac12|\theta|\sqrt I$, and
$g_\theta-\fzero=2a\sqrt{\fzero}+2b\sqrt{\fzero}+(a+b)^2$ integrated against $u$ contributes
$\theta\Ezero[uq]$, a term at most $2\|b\|_2=o(\theta)$ by Cauchy--Schwarz, and one at most
$\|a+b\|_2^2=O(\theta^2)$. The same assumption gives $\Ezero q=0$
\citep[Theorem~7.2]{vandervaart1998asymptotic}.

\emph{Step 2.} Applied to $\varphi_i$ and $\varphi_i\varphi_j$, bounded by
Assumption~\ref{ass:bounded}, Step 1 gives $\uone=\uzero+\theta\vB+o(\theta)$ with
$\vB_i=\Ezero[\tilde\varphi_iq]$ and $\Cone=\Czero+O(\theta)$; hence $\matF=2\Czero+O(\theta)$
and $\matF^{-1}=\tfrac12\Czero^{-1}+O(\theta)$ by Assumption~\ref{ass:moments},
$\vY=\theta\vB+o(\theta)$ and $\vK=\tfrac12\theta\Czero^{-1}\vB+o(\theta)\to\mathbf0$.

\emph{Step 3.} The drift identities of Section~\ref{sec:projection} give
$\Eone\Lam=\tfrac12\vY\transp\matF^{-1}\vY$; and since \eqref{eq:surrogate} makes $\Lam$ the
linear form $\vK\transp\vphi$ up to an additive constant,
$\Var_0\Lam=\vK\transp\Czero\vK=\vY\transp\matF^{-1}\Czero\matF^{-1}\vY$. Substituting
$\Czero=\tfrac12\matF-\tfrac12(\Cone-\Czero)$ and bounding the resulting second term in absolute
value by $\tfrac12\|\Cone-\Czero\|\,\|\matF^{-1}\|\,\vY\transp\matF^{-1}\vY$ gives
$\Var_0\Lam=\Eone\Lam\,(1+O(\theta))$.

\emph{Step 4.} With $a_\theta=\sup|\Lam|\le2\phimax\|\vK\|_1\to0$ and
$|e^{u}-1-u-u^2/2|\le|u|^3e^{|u|}/6$ at $u=t\Lam$,
$\Ezero e^{t\Lam}-1=t\Ezero\Lam+\tfrac12t^2\Ezero[(\Lam)^2]+O(a_\theta^3)$ uniformly on
$t\in[0,3]$; the remainder is $o(\Var_0\Lam)$ because
$\Var_0\Lam\ge\lambda_{\min}(\Czero)\|\vK\|^2$, and
$\Ezero[(\Lam)^2]=\Var_0\Lam(1+O(\|\vK\|^2))$. With Step 3 and $\Ezero\Lam=-\Eone\Lam$,
$\Ezero e^{t\Lam}-1=\tfrac12\Var_0\Lam\bigl(t^2-2t+\rho_\theta(t)\bigr)$ with
$\sup_{t\in[0,3]}|\rho_\theta(t)|\to0$, so for each $\eps\in(0,1)$ and $\theta$ small the bracket
is negative at $t=2-\eps$ and positive at $t=2+\eps$: the positive root of (a) satisfies
$|\thetaz-2|<\eps$, that is $\thetaz\to2$.

\emph{Step 5.} By Steps 2--4,
$\thetaz\Eone\Lam=(1+o(1))\vY\transp\matF^{-1}\vY
=\tfrac12\theta^2\vB\transp\Czero^{-1}\vB+o(\theta^2)$; dividing by the assumed expansion
$\Dkl(g_\theta\|\fzero)=\theta^2I/2+o(\theta^2)$ gives $\eff\to\kap$, and along the way
$\Jfunc{s}=\vY\transp\matF^{-1}\vY=\tfrac12\theta^2I\kap+o(\theta^2)$ in the normalisation
$\matF=\Czero+\Cone$.
Finally $\vB\transp\Czero^{-1}\vB=\Ezero[(\Pi q)^2]$ with
$\Pi q=\vB\transp\Czero^{-1}\tilde\vphi$ the $L^2(\fzero)$ projection of $q$ onto the span of
$\tilde\varphi_1,\dots,\tilde\varphi_s$, so $0\le\vB\transp\Czero^{-1}\vB\le\Ezero q^2=I$, with
equality on the right if and only if $q=\Pi q$ $\fzero$-almost surely.
\qed


\subsection{Proofs for Section~\ref{sec:plugin}}\label{app:plugin}

\subsubsection{Proof of Theorem~\ref{thm:plugin}}

\begin{proof}
\emph{(a)} The route is the delta method for a smooth function of sample moments, after
\citet[Ch.~3]{vandervaart1998asymptotic}; what is specific here is the Jacobian, which is read
off the normal system rather than off a closed form for $\vK$. Under
Assumption~\ref{ass:bounded} every coordinate of $\mathbf Z_0$ and of $\mathbf Z_1$ is bounded,
so all moments below are finite, and by the strong law and the central limit theorem
$\bar{\mathbf Z}_{0,\Ncal}\to\mathbf 0$ almost surely while the stacked vector
$\sqrt{\Ncal}(\bar{\mathbf Z}_{0,\Ncal},\bar{\mathbf Z}_{1,N_1})$, whose two blocks are
independent and satisfy $N_1/\Ncal\to\lambda$, converges to $\mathcal N(0,\mathbf V)$ with
$\mathbf V=\operatorname{diag}(\mathbf V_0,\lambda^{-1}\mathbf V_1)$, the second block absent in
mode~(M) and zero when $\lambda=\infty$. Up to the factor $\Ncal/(\Ncal-1)$,
$\hat{\Czero}-\Czero$ is the second block of $\bar{\mathbf Z}_{0,\Ncal}$ minus
$(\hat{\uzero}-\uzero)(\hat{\uzero}-\uzero)\transp=O_p(\Ncal^{-1})$, so the two have the same
first-order expansion; the same holds for $\hat{\Cone}-\Cone$ in mode~(E).

The map $\boldsymbol\eta\mapsto\matF^{-1}\vY$ is smooth in a neighbourhood of the population
value: $\matF\succeq\Czero\succ0$ there by Assumption~\ref{ass:moments}, and matrix
inversion is smooth on the positive-definite cone. Differentiating $\matF\vK=\vY$ gives
$\mathrm d\matF\,\vK+\matF\,\mathrm d\vK=\mathrm d\vY$, which is the first display of
\eqref{eq:plugin-jac}; the two mode-specific lines are the differentials of
$\vY=\uone-\uzero$ and $\matF=\Czero+\Cone$, with $\mathrm d\uone=\mathbf 0$,
$\mathrm d\Cone=\mathbf 0$ in mode~(M), where $(\uone,\Cone)$ is a supplied constant and not an
estimate. Differentiating $k_0=-\tfrac12\vK\transp(\uzero+\uone)$ gives the stated
$\mathrm dk_0$, so $(\hat{\vK},\hat k_0)$ is jointly asymptotically normal with covariance
$\boldsymbol\Sigma_{\vK,k_0}=\mathbf G_{(\vK,k_0)}\mathbf V\mathbf G_{(\vK,k_0)}\transp$. The
continuous-mapping theorem gives the almost-sure convergence and the delta method the limit law.
In mode~(E) the two blocks of $\mathbf V$ are independent, so
$\boldsymbol\Sigma=\boldsymbol\Sigma_0+\lambda^{-1}\boldsymbol\Sigma_1$ with
$\boldsymbol\Sigma_i$ the contribution of the $i$-th sample.

\emph{(b), Rayleigh quotient.} This part is a stationarity argument: $\vK$ maximises $R$, so the
first-order term is absent and the loss is the quadratic form at the maximiser. Write
$\hat{\vK}=\vK+\mathbf d$, $a=\mathbf d\transp\vY$, $b=\mathbf d\transp\matF\mathbf d$, and use
$\vK\transp\matF\vK=\vK\transp\vY=\Jfunc{s}$ and $\vK\transp\matF\mathbf d=\vY\transp\mathbf d$:
\[
R(\vK+\mathbf d)=\frac{(\Jfunc{s}+a)^2}{\Jfunc{s}+2a+b}
=\Jfunc{s}-\frac{\Jfunc{s}\,(b-a^2/\Jfunc{s})}{\Jfunc{s}+2a+b}
=\Jfunc{s}-\mathbf d\transp\Bigl(\matF-\frac{\vY\vY\transp}{\Jfunc{s}}\Bigr)\mathbf d\,
\bigl(1+O(\|\mathbf d\|)\bigr).
\]
The quadratic form is non-negative by the Cauchy--Schwarz inequality in the $\matF$-inner
product, $(\mathbf d\transp\vY)^2=(\mathbf d\transp\matF\vK)^2\le(\mathbf d\transp\matF\mathbf
d)(\vK\transp\matF\vK)$. With $\mathbf d=O_p(\Ncal^{-1/2})$ from~(a) the first claim follows, and,
provided the convergence is uniformly integrable,
$\Ncal\,\E[\Jfunc{s}-R(\hat{\vK})]\to
\operatorname{tr}((\matF-\vY\vY\transp/\Jfunc{s})\boldsymbol\Sigma)$.
The captured fraction $\kap=\vB\transp\Czero^{-1}\vB/I$ of \eqref{eq:kappa} is the maximum of the
same quotient with $(\matF,\vY)$ replaced by $(\Czero,\vB)$ and the whole divided by $I$, so the
display above holds verbatim about $\Czero^{-1}\vB$, with $I^{-1}$ multiplying both sides. It is a
statement about that quotient and not about an estimator: $\vB_i=\Ezero[(\vphi-\uzero)_i q]$
involves the score $q$ of Assumption~\ref{ass:local}, which the calibration sample does not
carry, so no plug-in $\hat\kap$ is defined here and no rate is claimed for one.

\emph{(b), efficiency.} Here the mechanism is the implicit-function theorem of
\citet[Ch.~3]{vandervaart1998asymptotic}, applied to the tilt equation and followed by the
product rule. For $\Lambda_{\vK,c}=c+\vK\transp\vphi$ set $G(\theta,\vK,c)=\Ezero e^{\theta\Lambda_{\vK,c}}-1$.
Under Assumption~\ref{ass:bounded} the exponent is bounded on compact sets of $(\theta,\vK,c)$,
so dominated convergence makes $G$ of class $C^\infty$ and permits differentiation under the
integral sign. At the root $\thetaz$ of \eqref{eq:tilt}, $\partial_\theta G=\tilde{\E}_0\Lam>0$,
because $\psi(\theta)=\log\Ezero e^{\theta\Lam}$ is strictly convex with $\psi(0)=\psi(\thetaz)=0$
and $\thetaz>0$ is its second zero. The implicit-function theorem therefore makes
$\thetaz(\vK,c)$ smooth, with
$\partial\thetaz/\partial(\vK,c)=-\thetaz\,(\tilde{\E}_0\vphi,1)/\tilde{\E}_0\Lam$, and
$\Eone\Lambda_{\vK,c}=c+\vK\transp\uone$ is linear in $(\vK,c)$; the product rule applied to
$e=\thetaz\,\Eone\Lambda/\Dkl(\fone\|\fzero)$ gives the gradient in \eqref{eq:plugin-eff}. It
vanishes if and only if $\tilde{\E}_0\vphi=\uone$ and $\tilde{\E}_0\Lam=\Eone\Lam$. When
$\thetaz\Lam=\ell$ the tilted law is $\fone$ and both hold. In general neither does: the normal
system matches $\Lam$ to the bounded score in the $\matF$-inner product
(Proposition~\ref{prop:ku1}), which does not match the tilted law to $\fone$. Since
$(\hat{\vK},\hat k_0)$ is jointly asymptotically normal by~(a) and $e$ is $C^2$ near
$(\vK,k_0)$ by the same implicit-function argument, a second-order Taylor expansion gives
\eqref{eq:plugin-eff}, whose leading term is $O_p(\Ncal^{-1/2})$ whenever $\nabla e\neq\mathbf 0$.

\emph{(c)} This part owns no imported mechanism: it is the elementary variance of a sample
variance, and Assumption~\ref{ass:bounded} is what makes the fourth moment it needs finite.
$(\hat{\Czero})_{ss}$ is the sample variance of $\varphi_s$; its variance is
$(\mu_4-\mu_2^2)/\Ncal+O(\Ncal^{-2})$, $\mu_k$ being the central moments of $\varphi_s$ under
$\fzero$, so the relative standard error is
$\sqrt{(\mu_4/\mu_2^2-1)/\Ncal}=\sqrt{(\gamma_4(\varphi_s)+2)/\Ncal}$ up to $O(\Ncal^{-1})$;
when $\mu_4=\mu_2^2$, which happens exactly for a two-point feature with equal masses, the
$O(\Ncal^{-2})$ term of the variance is all there is and the relative error is $O(\Ncal^{-1})$. Under
Assumption~\ref{ass:bounded}, $|\varphi_s|\le\phimax$ and $\mu_4\le\phimax^4<\infty$.

\emph{(d)} Both limits are the delta method of \citet[Ch.~3]{vandervaart1998asymptotic} applied
to the expansion of~(a); what is new is that $\hat\thetaz$ is not a function of
$(\hat{\vK},\hat k_0)$ alone, so the estimating equation carries its own sampling term.

For $\hat J$ the map is explicit: differentiating $\Jfunc{s}=\vY\transp\matF^{-1}\vY$ gives
$\mathrm dJ=2\vY\transp\matF^{-1}\mathrm d\vY-\vY\transp\matF^{-1}\,\mathrm d\matF\,\matF^{-1}\vY
=2\vK\transp\mathrm d\vY-\vK\transp\mathrm d\matF\,\vK$, and composing with
\eqref{eq:plugin-jac} and applying the delta method to
$\sqrt{\Ncal}(\hat{\boldsymbol\eta}-\boldsymbol\eta)\Rightarrow\mathcal N(0,\mathbf V)$ gives the
first limit, the last coordinate of $\mathbf Z_0$ not entering.

For $\hat\thetaz$ let
$G_{\Ncal}(\theta,\boldsymbol\eta)=\Ncal^{-1}\sum_{j\le\Ncal}e^{\theta\Lambda_{\boldsymbol\eta}(x_j)}-1$,
where $\Lambda_{\boldsymbol\eta}=k_0(\boldsymbol\eta)+\vK(\boldsymbol\eta)\transp\vphi$ is built
from the moment vector by \eqref{eq:normal-system} and \eqref{eq:surrogate}. By definition
$G_{\Ncal}(\hat\thetaz,\hat{\boldsymbol\eta})=0$, while
$G(\thetaz,\boldsymbol\eta)=\Ezero e^{\thetaz\Lam}-1=0$. Three facts make a first-order expansion
legitimate. First, $|\Lambda_{\boldsymbol\eta}|\le\phimax\|\vK(\boldsymbol\eta)\|_1+|k_0|$ is
bounded uniformly on a neighbourhood of $\boldsymbol\eta$ by Assumption~\ref{ass:bounded} and the
smoothness from~(a), so $e^{\theta\Lambda_{\boldsymbol\eta}(x)}$ and its derivatives in
$(\theta,\boldsymbol\eta)$ are bounded there uniformly in $x$, and $G_{\Ncal}$ and its first two
derivatives converge uniformly on it to those of $G$. Second, $\partial_\theta
G(\thetaz,\boldsymbol\eta)=\tilde{\E}_0\Lam>0$ as in~(b), so for large $\Ncal$ the empirical
equation has a root $\hat\thetaz\to\thetaz$ in probability and $\partial_\theta G_{\Ncal}$ is
bounded away from zero at it. Third,
$G_{\Ncal}(\thetaz,\boldsymbol\eta)=\Ncal^{-1}\sum_{j\le\Ncal}(e^{\thetaz\Lam(x_j)}-1)$
is the sample mean of the last coordinate of $\mathbf Z_0$, which has mean zero under $\fzero$ by
the definition of $\thetaz$ and finite variance $\Var_0e^{\thetaz\Lam}$ by
Assumption~\ref{ass:bounded}. Expanding $G_{\Ncal}$ about $(\thetaz,\boldsymbol\eta)$ and using
$\partial_{\boldsymbol\eta}G(\thetaz,\boldsymbol\eta)
=\thetaz\,\tilde{\E}_0[\partial_{\boldsymbol\eta}\Lambda_{\boldsymbol\eta}]
=\thetaz\,(\tilde{\E}_0\vphi,1)\,\mathbf G_{(\vK,k_0)}$,
\[
\hat\thetaz-\thetaz
=-\frac{1}{\tilde{\E}_0\Lam}\Bigl[
\frac1{\Ncal}\sum_{j\le\Ncal}\bigl(e^{\thetaz\Lam(x_j)}-1\bigr)
+\thetaz\,(\tilde{\E}_0\vphi,1)\,\mathbf G_{(\vK,k_0)}(\hat{\boldsymbol\eta}-\boldsymbol\eta)
\Bigr]+o_p(\Ncal^{-1/2}),
\]
which is the inner product of $\nabla\thetaz$ of \eqref{eq:plugin-an} with the stacked sample
mean of $(\mathbf Z_0,\mathbf Z_1)$. Both terms are averages over the same calibration sample,
so their covariance is carried by $\mathbf V_0$ and not discarded; the labelled sample enters
only through $\hat{\boldsymbol\eta}$ and is independent. Applying the central limit theorem to
the stacked mean gives
$\sqrt{\Ncal}(\hat\thetaz-\thetaz)\Rightarrow\mathcal N(0,\nabla\thetaz\transp\mathbf
V\nabla\thetaz)$, and expanding that quadratic form in the blocks of $\mathbf V$ gives the
display of part~(d), the $\lambda^{-1}$ term being absent in mode~(M), where
$\hat{\boldsymbol\eta}$ has no labelled-sample block.
\end{proof}

\end{document}